\documentclass[sigconf]{acmart}
\usepackage{caption}
\allowdisplaybreaks[4]
\usepackage{amsmath}

\usepackage{amssymb}
\usepackage[normalem]{ulem}
\usepackage{multirow}
\usepackage{graphicx}
\usepackage{subfigure}
\usepackage{bbding}
\usepackage{threeparttable}
\usepackage{colortbl}
\usepackage{tablefootnote}
\useunder{\uline}{\ul}{}
\newtheorem{theorem}{Theorem}
\newtheorem{lemma}{Lemma}
\usepackage[linesnumbered, ruled]{algorithm2e}
\newcommand{\upcite}[1]{\textsuperscript{\textsuperscript{\cite{#1}}}}

\AtBeginDocument{%
  \providecommand\BibTeX{{%
    \normalfont B\kern-0.5em{\scshape i\kern-0.25em b}\kern-0.8em\TeX}}}

\begin{document}

\title{Hyperbolic Multiplex Network Embedding with Maps of Random Walk}

\author{Peiyuan Sun}
\email{sunpy@act.buaa.edu.cn}
\affiliation{%
  \institution{Beijing Advanced Institution on Big Data and Brain Computing, Beihang University}
  \city{Beijing}
  \state{China}
  \postcode{43017-6221}
}

\author{Jianxin Li}
\email{lijx@act.buaa.edu.cn}
\affiliation{%
  \institution{Beijing Advanced Institution on Big Data and Brain Computing, Beihang University}
  \city{Beijing}
  \state{China}
  \postcode{43017-6221}
}

\renewcommand{\shortauthors}{Peiyuan Sun, et al.}

\begin{abstract}
  Recent research on network embedding in hyperbolic space have proven successful in several applications. However, nodes in real world networks tend to interact through several distinct channels. Simple aggregation or ignorance of this multiplexity will lead to misleading results. On the other hand, there exists redundant information between different interaction patterns between nodes. Recent research reveals the analogy between the community structure and the hyperbolic coordinate. To learn each node's effective embedding representation while reducing the redundancy of multiplex network, we then propose a unified framework combing multiplex network hyperbolic embedding and multiplex community detection. The intuitive rationale is that high order node embedding approach is expected to alleviate the observed network's sparse and noisy structure which will benefit the community detection task. On the contrary, the improved community structure will also guide the node embedding task. To incorporate the common features between channels while preserving unique features, a random walk approach which traversing in latent multiplex hyperbolic space is proposed to detect the community across channels and bridge the connection between node embedding and community detection. The proposed framework is evaluated on several network tasks using different real world dataset. The results demonstrates that our framework is effective and efficiency compared with state-of-the-art approaches.
\end{abstract}





\maketitle


\section{I\lowercase{ntroduction}}
Network embedding learns for each node a coordinate in latent space. It has proven in several important network tasks such as node classification, community detection, etc. Most current network embedding approaches\cite{tang2015line}\cite{perozzi2014deepwalk} assume a single type relationship between nodes. However, in real world, nodes tend to interact with each other in distinct ways. Simple aggregation or ignorance of this multiplexity will lead to misleading results\cite{boccaletti2014structure}\cite{kivela2014multilayer}.

Recently, several network embedding approaches for multiplex networks have been proposed\cite{zhang2018scalable}\cite{matsuno2018mell}\cite{cen2019representation}\cite{shi2018mvn2vec}\cite{qu2017attention}\cite{liu2017principled}. Most of them aim to incorporate the common features while preserving distinct characters of each layer.

In this work we follow this idea by employing a novel multiplex community approach through which the common features between different layers are shared. Several community analysis on the multiplex network have been proposed. The key assumption of these work is that nodes share the same block structure over the multiple layers, but the class connection probabilities may vary across layers\cite{taylor2016enhanced}\cite{han2015consistent}. \cite{stanley2016clustering} extends this assumption by that the multiplex network has a group structure of layers. All layers in the same group share the same community assignments. Another lines of work employ the information tools to quantify the similarities between layers\cite{de2015structural}\cite{de2016spectral}. Then hierarchical cluster analysis can be performed by the layers distance metric. \cite{kuncheva2015community} defined a supra-adjacency matrix in which the transition probabilities are locally adapted. We extend this idea in a further step by means of modular flow identification. The intuitive idea is that if we guide a random walker to traverse in multiple latent hyperbolic spaces corresponding to each layer. Then the multiplex modular means the random walker stayed for a relatively long time. We argue that multiplex modular reveals the local similarity structure between layers. This is the most significant difference compared with all previous methods.

Most current embedding methods compute the latent node coordinates in Euclidean space. However, according to recent researches\cite{papadopoulos2012popularity}\cite{krioukov2010hyperbolic}, hyperbolic space is considered as a more reasonable latent space than Euclidean space. As scale-free and high clustering coefficient, the two most fundamental and ubiquitous properties, emerge naturally in hyperbolic space. Several hyperbolic embedding approaches have also been proposed\cite{nickel2017poincare}\cite{papadopoulos2015network}\cite{papadopoulos2015networkcommon}. However, none of these methods consider the multiplexity. The second rationale behind the hyperbolic space comes from recent findings on the analogy between hyperbolic embedding and community structure\cite{faqeeh2018characterizing}\cite{wang2017community}. We then are inspired by these findings and extend them to multiplex community scenario. This bridges the gap between single layer's unique structure and similarities between layers. More specifically, we force the node's hyperbolic coordinates within one multiplex community to be close to each other while preserving each layer's own structure. One maybe tempted to perform hyperbolic embedding after the multiplex community detection. However, as several previous work proved\cite{cavallari2017learning}\cite{wang2017community}\cite{lai2017prune}, these two tasks benefit each other reciprocally. We then propose a unified framework for multiplex community detection and hyperbolic embedding learning for each node.
In the start step, preliminary hyperbolic embeddings are obtained by preserving one and two order proximity. The modular flow approach is then employed to perform the multiplex community detection. However, compared with its origin idea, we modified the random walker to traverse in the latent hyperbolic space with a teleportation jump to neighbor node in any layer. A community coherence regularizer is then used to confine the coordinates within one community. This coherence is inspired by the Kuramoto model. Since the hyperbolic space in this work employs the polar coordinates. The angular coordinates are forced to be close further. These three steps are iterated until a local minima result is obtained.

We summarize our contribution as follows:
\begin{itemize}
    \item We propose a unified approach to simultaneously perform the network embedding and multiplex community detection task.
    \item We perform the convergence analysis on the proposed approach.
    \item We evaluate the approach on multiple real-world dataset. The result demonstrates that our approach excels in several network analysis tasks with regard to state-of-the-art baselines.
\end{itemize}


\section{R\lowercase{andom} H\lowercase{yperbolic} D\lowercase{isk} M\lowercase{odel}}
A random hyperbolic disk model\cite{gugelmann2012random} constructs a network $(V, E)$ with parameters: $\alpha, C, n$. For a scale free network whose node degree distribution follows a power law with exponent $\beta=2\alpha+1 (\alpha > \frac{1}{2}\text{ for most real networks})$. The node set of network is $V=\{1, 2,\cdots, n\}$. The disk diameter is set to $R=2{\rm logn} + C$. Each node is with two coordinates $(r_i, \theta_i)$, where the density function of $r$:
$$p(r)=\alpha\frac{{\rm sinh}(\alpha r)}{{\rm cosh}(\alpha R)-1}$$.
And angular coordinate $\theta$ draws uniformly in $[0, 2\pi]$.
The edge between node pair $u, v$ is draw from the probability function:
$$p((u, v)\in E)=\left(1+{\rm exp}\left(\frac{(d(u, v)-R)}{2T}\right)\right)^{-1}$$,
where the hyperbolic distance function
$$d(u, v)={\rm arcosh}({\rm cosh}(r_u){\rm cosh}(r_v)-{\rm sinh}(r_u){\rm sinh}(r_v){\rm cos}(\theta_u-\theta_v))$$. 


\section{C\lowercase{ommunity} C\lowercase{onductance}}

We first clarify the notations used in this work. Then two theorems based on the random hyperbolic graph model are proposed to guide the node's latent hyperbolic coordinates layout. We denote the graph as $G=(V, E)$, where $V$ and $E$ corresponds to the node set and edge set respectively.

\textbf{Volume, cut and conductance}. Recall that for $S\subseteq V$, ${\rm vol}(S)=\sum\limits_{v\in S}{\rm deg}(v)$. The cut induced by $S$ is denoted as $\partial S=\{uv, u\in S, v\in\bar{S}\}$ where $\bar{S}$ denotes the complementary set of $S$. The conductance of $S$ is defined as:

\begin{equation}
h(S)=\frac{|\partial S|}{min\left\{{\rm vol}(S), {\rm vol}(\bar{S})\right\}}
\end{equation}

For later proof we also propose the notation of relative conductance of two set $A$ and $B$ as:
\begin{equation}
R(A, B)=\frac{|E(A, B)|}{min\left\{{\rm vol}(A), {\rm vol}(B)\right\}}
\end{equation} which measures the closeness of two communities node set within.

\begin{theorem}\label{theorem1}
For three communities $A$, $B$ and $C$ generated by random hyperbolic graph model, if the communities are arranged as the sequential order, then the following inequality holds:
\begin{equation}\label{eq1intheorem1}
R(A,B)>R(A,C)
\end{equation}
\end{theorem}

\begin{proof}
assume $A$, $B$ and $C$ corresponds to three sector areas with angular $\Delta\theta_A$, $\Delta\theta_B$ and $\Delta\theta_C$ respectively. Figure~\ref{SectorOrder} illustrates one possible circular order of sector $A,B$ an $C$.
Then inequality \ref{eq1intheorem1} is defined as:

\begin{figure}
\centering
\includegraphics[width=0.2\textwidth]{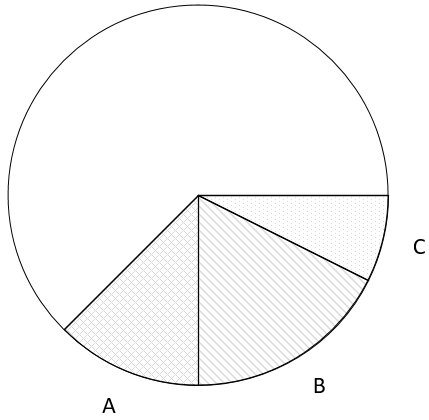}
\caption{Illustration for sector area A, B and C in Poincare Disk.}
\label{SectorOrder}
\setlength{\abovecaptionskip}{5pt}
\setlength{\belowcaptionskip}{5pt}
\end{figure}

\begin{equation}\label{eq15}
\frac{E(A,B)}{min\left\{{\rm vol}(A), {\rm vol}(B)\right\}}>\frac{E(A,C)}{min\left\{{\rm vol}(A), {\rm vol}(C)\right\}}
\end{equation}

First,since nodes distribute uniformly in the hyperbolic disk. The expected number of nodes in sector $S$ with angular $\Delta\theta$ is:

\begin{equation}\label{eq16}
\mathbb{E}(|S|)=\frac{\Delta\theta}{2\pi}N
\end{equation}

We then estimate the expected number of edges between two sector areas.
\begin{equation}\label{eq17}
\mathbb{E}(|E(S,\bar{S})|)=|S||\bar{S}|\cdot\mathbb{P}\left(u\in S, v\in \bar{S}, u\sim v\right)
\end{equation}

We divide the above probability into two cases depend on whether $r_u+r_v<R$ and assume ${\rm vol}(C)<{\rm vol}(A)<{\rm vol}(B)$ without loss of generality.

For $r_u+r_v<R$,
\begin{equation}\label{eq17}
\begin{split}
P_1 &= \iint\limits_{r_u+r_v<R}f(r_u)f(r_v)dr_udr_v\\
&=\frac{\alpha}{2C^2(\alpha,R)}R\cdot {\rm sinh}(\alpha R)+\frac{{\rm cosh}(\alpha R)}{C^2(\alpha, R)}-\frac{1}{C^2(\alpha, R)}\\
&=O(R)e^{-\alpha R}
\end{split}
\end{equation}

For $r_u+r_v\ge R$,
\begin{equation}\label{eq18}
\begin{split}
P_2 &= \frac{(\theta_R-\theta_B)^2}{2\theta_A\theta_C}\iint\limits_{r_u+r_v\ge R}f(r_u)f(r_v)dr_u dr_v\\
&= \frac{\alpha e^R}{2\theta_A\theta_C C^2(\alpha,R)}\cdot O(R)\cdot e^{-\alpha R}
\end{split}
\end{equation}
\end{proof}

Theorem~\ref{theorem1} indicates that the communities can be distributed on the hyperbolic disk according to their relative conductance correlation.

\begin{lemma}\label{measureRHD}(Lemma 3.2 in \cite{gugelmann2012random})
If $0\le r< R$, then $\mu(B_0(r))=e^{-\alpha(R-r)}(1+o(1))$
\end{lemma}

\begin{lemma}\label{PPP}(\cite{friedrich2018on})
We consider to use the Poisson point process to describe the nodes distribution in random hyperbolic disk, which states the probability that there exists at least one point in some area is: \begin{equation}
P(\exists v\in S)=1-{\rm exp}\left(-n\cdot\mu(S)\right)
\end{equation}
\end{lemma}

\begin{lemma}\label{HighDegreeNodeProb}
We consider nodes with radial coordinate less that $\frac{R}{2}$ as the core node set of the random hyperbolic disk. For a graph with $N$ nodes, if we partition it into at most $K$ sectors with equal angles, then with probability $p=1-{\rm exp}\left(-n^{1-\alpha}/K\right)$, each sector has at least one core node.
\end{lemma}

In \cite{blasius2018efficient}, the author provides an approximate expression for the relative angle between two core nodes. However, in the following theorem we will show that the community based core nodes circular ordering estimation is more robust than common neighbor based one.
\begin{theorem}\label{better}
The expectation number of two core nodes for a scale free network with $n$ nodes is $O(n^{2\alpha(1-\alpha)})$. While the relative conductance based core nodes circular ordering is $O(n^{2(1-\alpha)}{\rm log}n)$. As the expectation of relative conductance is greater than common neighbor, by employing Chernoff bound, the confidence interval is sharper than the latter.
\end{theorem}

%
%
%
%

It seems that our proposed relative community conductance is similar to the Community Intimacy (CI) in \cite{wang2016hyperbolic}. Actually, it can be easily verified that the following inequality holds. 

\begin{equation}
\label{RBound}
\begin{split}
\frac{1}{2}CI(A, B)\le R(A,B)\le CI(A,B)
\end{split}
\end{equation}

In the next section, however we will show the hidden relationship between Community Conductance and the Map Equation. And the analogy between Map Equation and Poincare Disk embedding is also revealed.

\section{M\lowercase{ap} E\lowercase{quation} R\lowercase{eveals} B\lowercase{est} A\lowercase{nalogy} \lowercase{with} H\lowercase{yperbolic} G\lowercase{eometry}}

\subsection{Map Equation exhibits best resolution limit}
\begin{theorem}\label{infomapResolutionLimit}
For networks with total number of $L$ links, the number of modules that infomap algorithm detects has intrinsic scale as following:
\begin{equation}
m^*=O\left(L/{\rm ln}L\right)
\end{equation}
\end{theorem}

\begin{proof}
Recall that the original infomap algorithm is to minimize the following equation:
\begin{equation}
\begin{split}
L(M)&=q_{\curvearrowright}H(Q)+\sum\limits_{i=1}^Mp_{\circlearrowleft}^iH(P_i)\\
&=\left(\sum_{i=1}^mq_{i\curvearrowright}\right){\rm log}\left(\sum_{i=1}^mq_{i\curvearrowright}\right)-2\sum_{i=1}^mq_{i\curvearrowright}{\rm log}q_{i\curvearrowright}\\
&+\sum_{i=1}^m\left(q_{i\curvearrowright}+\sum_{\alpha\in i}p_{\alpha}\right){\rm log}\left(q_{i\curvearrowright}+\sum_{\alpha\in i}p_{\alpha}\right)
\end{split}
\end{equation}
where $M$ denotes a partition of the network. For a undirected unweighted network, we can easily get the following formulas.

\begin{equation}
\begin{split}
q_{i\curvearrowright}&=\sum\limits_{\alpha\in i, \beta\notin j}p_\alpha p_{ij}\\
&=\sum\limits_{\alpha\in i, \beta\notin j}\frac{d_\alpha}{d}\frac{A_{\alpha\beta}}{d_\alpha}\\
&=\frac{{\rm cut}(i)}{d}
\end{split}
\end{equation}

We plug in the above formulas and denotes $vol(i)$ as the sum of nodes' degree in partition $i$ and transform the original map equation as following:
\begin{equation}
\begin{split}
L(M)&=\left(\sum_{i=1}^m\frac{{\rm cut}(i)}{d}\right){\rm log}\left(\sum_{i=1}^m\frac{{\rm cut}(i)}{d}\right)\\
&-2\sum_{i=1}^m\frac{{\rm cut}(i)}{d}{\rm log}\frac{{\rm cut}(i)}{d}\\
&+\sum_{i=1}^m\left(\frac{{\rm cut}(i)}{d}+\frac{{\rm vol}(i)}{d}\right){\rm log}\left(\frac{{\rm cut}(i)}{d}+\frac{{\rm vol}(i)}{d}\right)
\end{split}
\end{equation}

We first compute the derivative of $L(M)$ with respect to $cut(i)$:
\begin{equation}
\frac{\partial L}{\partial {\rm cut}(i)} = \frac{1}{d}{\rm log}\frac{\sum_{i=1}^m{\rm cut}(i)\left({\rm cut}(i)+{\rm vol}(i)\right)}{{\rm cut}(i)^2}
\end{equation}

which is is always greater than 0 which states that the value of map equation will increase as the increasing of total cut edges between modulars. On the other hand, we can employ the lagrange multiplier method to compute the maxima of $L(M)$ with regards to $vol(i)$.

\begin{equation}
\begin{split}
L'&=L(M)+\lambda(\sum_{i=1}^m{\rm vol}(i)-d)
\end{split}
\end{equation}

We perform derivative on map equation with the total volumes constraints as following:

\begin{equation}
\begin{split}
\frac{L'}{{\rm vol}(i)}&=\frac{1}{d}{\rm log}\frac{{\rm cut}(i)+{\rm vol}(i)}{d}+\frac{1}{d}+\lambda\\
\end{split}
\end{equation}

We let the above derivative equals to 0 and get the conclusion that the map equation minimizes when each modular's volume and cut is identical.
\begin{equation}
\begin{split}
{\rm cut}(i)+{\rm vol}(i)=\frac{\sum_{i=1}^m{\rm cut}(i)+d}{m}
\end{split}
\end{equation}

\begin{figure}
\centering
\includegraphics[width=0.2\textwidth]{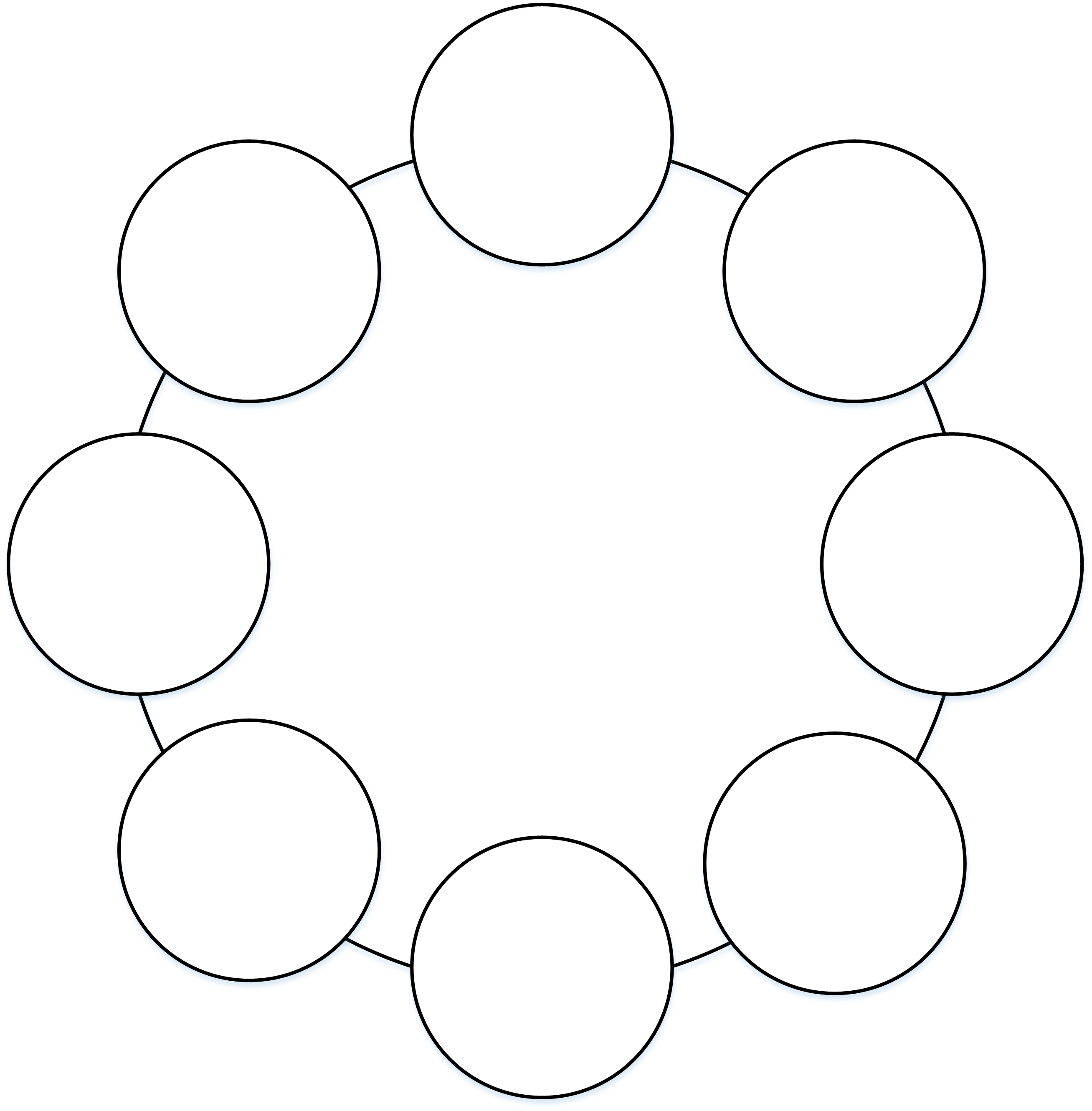}
\caption{Illustration for ring structure of modulars}
\label{RingCluster}
\end{figure}

For clarity and simplicity of mathematical analysis, we assume each $cut(i)$ equals 2 which means that the modular structures as a ring (as illustrated in Figure~\ref{RingCluster}). Then the map equation transforms to the following forms under the above assumption:
\begin{equation}
\begin{split}
L(M)&=\frac{2m}{d}{\rm log}\frac{2m}{d}-\frac{4m}{d}{\rm log}\frac{2}{d}+\frac{2m+d}{d}{\rm log}\frac{2m+d}{dm}
\end{split}
\end{equation}

Then we compute the minima by performing derivativation of the above formula. However we get an transcendental equation.
\begin{equation}
{\rm log}(m+L)=\frac{L}{m}-1
\end{equation}
which can only be approximation by numerical methods. While after observation, we can get a well approximated analytical solution:
\begin{equation}
m=\frac{L}{{\rm ln}L}
\end{equation}
which has constant distortion with the optimal solution since the error goes to constant 1 when $L$ goes to infinity

\begin{equation}
error = {\rm ln}\left(1+\frac{1}{{\rm ln}L}\right) + 1 \longrightarrow 1
\end{equation}

which ends the proof.

\end{proof}

\subsection{Map Equation performs edge balanced partition}
\begin{theorem}\label{infomapBalanceEdgePartitionEq}
The map equation partition objective function is equivalent to $(k, \nu)$-balanced edge partition in which $k$ denotes the number of clusters and $\nu=\mathop{max}\{\alpha_i\}$ equals to the maximal conductance of partitioned clusters.
\end{theorem}
\begin{proof}
Recall that the $(k, \nu)$-balanced edge partition is formulated to the following objective function and constraints:
\begin{equation}
\begin{aligned}
\text{minimize}\quad &\sum\limits_{u\in V}\left(d_{in}(u)-\sum\limits_{i=1}^k z_{u,i}\left(\sum\limits_{v\in N_{in}(u)x_{(v,u),i}}\right)\right)\\
\text{subject to}\quad &\sum\limits_{u\in V} d_{in}(u)z_{u,i}\le (1+\nu)\frac{m}{k},i\in[k]\\
&\sum\limits_{i=1}^k x_{e,i}=1, e\in E\\
&\sum\limits_{i=1}^k z_{u,i}=1, u\in V\\
& x\in\{0, 1\}^{m \times k}\\
& z\in\{0, 1\}^{n \times k}
\end{aligned}
\end{equation}
in which, $x_{e, i}=1$ if each edge $e$ is assigned to cluster $i$. And the second and fourth constrains assure that each edge is assigned to exactly one cluster. Then the objective function is equivalent to minimize the total number of edges across different clusters. The first constraint condition controls the balance of edge partition. According to the analysis in Theorem~\ref{infomapResolutionLimit}, the map equation minimizes when the sum of cut and volume of each cluster $i$ is equal:
\begin{equation}
C_i+V_i=\frac{\sum\limits_{i=1}^K (C_i+V_i)}{K}
\end{equation}
Denote the conductance of cluster $i$ is $\alpha_i$:
\begin{equation}
\alpha_i = \frac{C_i}{V_i}
\end{equation}
Then:
\begin{equation}
\begin{aligned}
C_i+V_i&=\frac{\sum\limits_{i=1}^K (\alpha_i V_i + V_i)}{K}\\
V_i&=\frac{\sum\limits_{i=1}^K ((1 + \alpha_i) V_i}{K}-\frac{\alpha_i K V_i}{K}
\end{aligned}
\end{equation}
Then map equation is equivalent to $(k, \nu)$-balanced edge partition if we set $\nu=\mathop{max}\limits_{i}\{\alpha_i\}$.
\end{proof}

It should be noted that the resolution limit of modularity\cite{fortunato2007resolution} is $O(\sqrt{L})$ while for stochastic block model with minimum description length principle (SBM-MDL) the detection limit for the size of block is $\sqrt{N}$\cite{peixoto2013parsimonious} . Then we can get the conclusion that the infomap algorithm excels modularity with regard to the resolution limit. It is well known that the community detection algorithms can be classified into three main categories:(1) null models include modularity and Louvain etc, (2) block models include SBM, DC-SBM, SBM-MDL, DCSBM-MDL etc, (3) flow models include Infomap etc. Resolution limit of the first two categories have been throughly studied. However the detectability of infomap algorithm is still unclear. In \cite{kawamoto2015estimating}, the author provides an condition relates the module volume and the total cut size. In our work, we will give a more clear and intrinsic resolution limit expression. The resolution limits of the most clustering algorithms are summarized in Table~\ref{resolutionLimitSummary}. In Table~\ref{ResolutionLimit}, we report the number of modules detected by Louvain, Infomap and SBM based clustering algorithms on networks from different domains and predicted by our resolution limit $O(L/{\rm log}L)$. The theoretical prediction versus the empirical community detection result is also illustrated in Figure~\ref{LlnL} for clarity. Details of the dataset employed in this table are listed in section~\ref{datasetSection}.

\begin{table}[]
\small
\begin{tabular}{|c|c|c|c|c|}
\hline
algorithm        & Modularity\upcite{fortunato2007resolution} & SBM-DL\upcite{peixoto2013parsimonious}    & N-SBM\upcite{peixoto2014hierarchical}   & Infomap \\ \hline
resolution & $\sqrt{L}$ & $\sqrt{N}$ &$N/{\rm ln}N$  & $L/{\rm ln}L$ \\ \hline
\end{tabular}
\caption{Summary on resolution limit of clustering algorithms. Resolution limit of Infomap is the new result in this work.}
\label{resolutionLimitSummary}
\end{table}

\begin{table}[]
\small
\begin{tabular}{|c|c|c|c|c|c|c|c|c|}
\hline
Dataset                         & Layer & N     & E     & L & SL & DCL & I & L/lnL \\ \hline
\multirow{2}{*}{Router}         & IPv4  & 37542 & 95019 & 31      & 239    & 176      & 1625    & 5742   \\ \cline{2-9}
                                & IPv6  & 5143  & 13446 & 19      & 86     & 79       & 418     & 980    \\ \hline
\multirow{3}{*}{C.E.}      & 1     & 248   & 514   & 9       & 15     & 19       & 29      & 57     \\ \cline{2-9}
                                & 2     & 258   & 887   & 9       & 26     & 21       & 23      & 91     \\ \cline{2-9}
                                & 3     & 278   & 1703  & 7       & 34     & 30       & 11      & 159    \\ \hline
\multirow{2}{*}{D.M.} & 1     & 752   & 1808  & 17      & 39     & 24       & 70      & 167    \\ \cline{2-9}
                                & 2     & 633   & 1343  & 17      & 41     & 20       & 68      & 129    \\ \hline
\multirow{6}{*}{arXiv}          & 1     & 1537  & 3935  & 32      & 58     & 53       & 130     & 330    \\ \cline{2-9}
                                & 2     & 2121  & 5473  & 35      & 72     & 61       & 190     & 438    \\ \cline{2-9}
                                & 3     & 129   & 273   & 10      & 13     & 12       & 17      & 34     \\ \cline{2-9}
                                & 4     & 3669  & 11969 & 46      & 113    & 94       & 290     & 884    \\ \cline{2-9}
                                & 5     & 608   & 1664  & 23      & 34     & 28       & 61      & 156    \\ \cline{2-9}
                                & 6     & 336   & 945   & 17      & 29     & 24       & 38      & 96     \\ \hline
\multirow{2}{*}{S.P.}      & 1     & 751   & 1213  & 21      & 30     & 29       & 86      & 118    \\ \cline{2-9}
                                & 2     & 182   & 246   & 13      & 14     & 9        & 28      & 31     \\ \hline
\multirow{2}{*}{Rattus}         & 1     & 1866  & 2657  & 32      & 29     & 32       & 129     & 234    \\ \cline{2-9}
                                & 2     & 529   & 647   & 20      & 18     & 18       & 61      & 69     \\ \hline
\end{tabular}
\caption{We report the number of modules detected in 17 networks from 6 different domains for Louvain, Infomap, SBM-MDL and DCSBM-MDL. In the last column  we report the theoretical resolution limit for infomap algorithm. 'L' denotes Louvain algorithm. 'SL' denotes SBM-MDL algorithm. 'DCL' denotes DCSBM-MDL algorithm. 'I' denotes Infomap algorithm.}
\label{ResolutionLimit}
\end{table}

\begin{figure}
\centering
\includegraphics[width=0.48\textwidth]{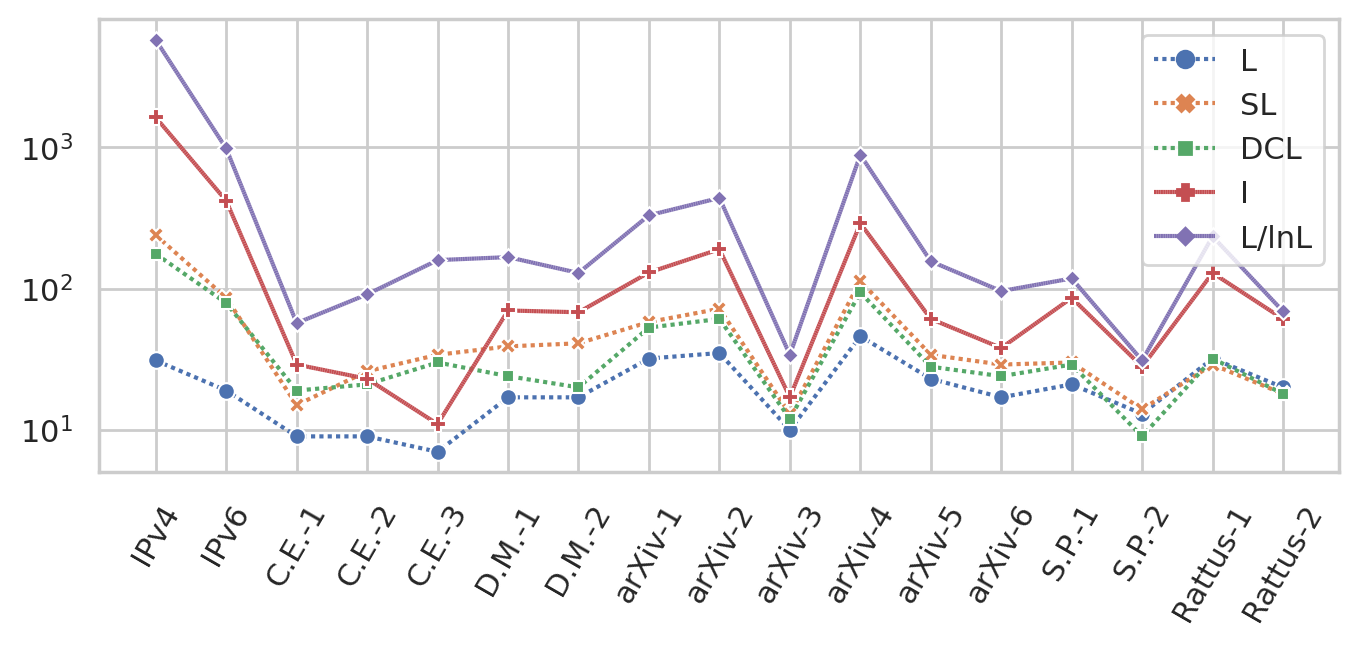}
\caption{Theoretical versus empirical resolution limit}
\label{LlnL}
\end{figure}

Although the analogy between hyperbolic embedding and community structure has been studied in previous research\cite{faqeeh2018characterizing}\cite{wang2016hyperbolic}\cite{wang2017community}. It is still not clear to what extent the detected community structure reveal the latent hyperbolic geometry of networks. In \cite{wang2016hyperbolic}\cite{wang2017community} modularity based algorithm is employed. While in \cite{faqeeh2018characterizing} Louvain and Infomap are used to characterize the module structure. In this work, we first investigate the coherence of hyperbolic geometry coordinates in the same community for different clustering algorithms. To this end, we compute the node coordinate coherence defined by euqation~\ref{coherenceEquation} following\cite{faqeeh2018characterizing}. 

\begin{equation}\label{coherenceEquation}
{\xi_g}=\frac{1}{n_g}\sqrt{\left\{\sum\limits_{j=1}^N\left(\delta_{\sigma_j,g}{\rm cos}\theta_j\right)\right\}^2 + \left\{\sum\limits_{j=1}^N\left(\delta_{\sigma_j,g}{\rm sin}\theta_j\right)\right\}^2}
\end{equation}

For each network we take the average for 100 instances of random hyperbolic model. The result is reported in Table~\ref{CoherenceVSclustering}. For Infomap we rely on the open implementation\footnote{https://www.mapequation.org/code.html}. For modularity, we employed the implementation in the open library of \verb|NetworkX|. For SBM and DC-SBM, the code was taken from the open project\footnote{https://github.com/funket/pysbm}. For SBM-MDL and DCSBM-MDL, the implementation relies on the open project \verb|graph-tools|. From Table~\ref{CoherenceVSclustering}, it can be found that Infomap algorithm excels all other clustering algorithms at all network scales. Figure~\ref{ClusteringIllustration} illustrates the clustering coherence for each clustering algorithm. The nodes rendered in the same color belongs to the same module. While modularity based algorithm and DCSBM-DL algorithm won two and four second places respectively. Furthermore, DCSBM-MDL performs better as the network scale increases. And modularity algorithm performs best at very small network and decays as the scale increases. SBM and DC-SBM perform worst due to a apriori block number needs to be specified which is however difficult. Intuitively, the best coherence can be maximized when the found module is dense and relative small. In DCSBM-DL, by employing the minimum description length principle, the maximum number of detectable blocks scales as $\sqrt{N}$. For sparse network (common in real networks), the number of edge is as the same scale of number of nodes. Then we can get the conclusion that Infomap algorithm has the best detectability of block structures among modularity, SBM, DCSBM and DCSBM-DL clustering algorithms. Furthermore, according to Theorem~\ref{infomapBalanceEdgePartitionEq}, infomap is equivalent to an edge balanced partition algorithm. This has been found to be more useful than conventional vertex based partition algorithm since most real networks are scale-free\cite{zhang2017graph}\cite{yang2017hypergraph}. One subtle difference between edge based partition and vertex based partition can be revealed by how to partition the core part of network. Since the high degree nodes will almost connect with each other. The vertex based partition algorithm will tend to enclose the high degree nodes as a module which will cause performance bottleneck in certain scenarios. While edge based partition instead prefers to cut the core nodes into several modules. In Figure~\ref{ClusteringIllustration}, this phenomenon can be found by comparing Infomap, modularity and SBM based algorithms. It is unclear whether modularity and SBM based algorithms are edge based or vertex based. However, by inspecting their resolution limit, we can classify them according to whether edge number or vertex number appeared in the expression. We since conjecture that modularity also belongs to the edge based partition algorithm. While SBM based partition algorithms belong to the vertex based partition algorithm.

\begin{table}[]
\small
\begin{tabular}{|c|c|c|c|c|c|c|}
\hline
N    & Infomap         & mod   & sbm    & sbm-dl  & dcsbm & dcsbm-dl      \\ \hline
100  & \textbf{0.9496} & {\ul 0.9260} & 0.8097 & 0.8497  & 0.8388 & 0.8321       \\ \hline
200  & \textbf{0.9784} & {\ul 0.9409} & 0.5688 & 0.9105  & 0.6857 & 0.8738       \\ \hline
500  & \textbf{0.9906} & 0.9516       & 0.2688 & 0.9554  & 0.3234 & {\ul 0.9671} \\ \hline
1000 & \textbf{0.9907} & 0.9463       & 0.1712 & 0.9485  & 0.2059 & {\ul 0.9772} \\ \hline
5000 & \textbf{0.9980} & 0.9439       & 0.0805 & 0.9485  & 0.0808 & {\ul 0.9768} \\ \hline
10000 & \textbf{0.9670} & 0.9010       & 0.0518 & 0.9384  & 0.0511 & {\ul 0.9567} \\ \hline
\end{tabular}
\caption{Coherences for different clustering algorithm and different scale network generated by random hyperbolic disk model. The best result is in bold font. While the second-best is underlined.}
\label{CoherenceVSclustering}
\end{table}

\begin{figure}[htbp]
    \centering
    \subfigure[Infomap]{
        \begin{minipage}[t]{0.3\linewidth}
        \centering
        \includegraphics[width=1.1in]{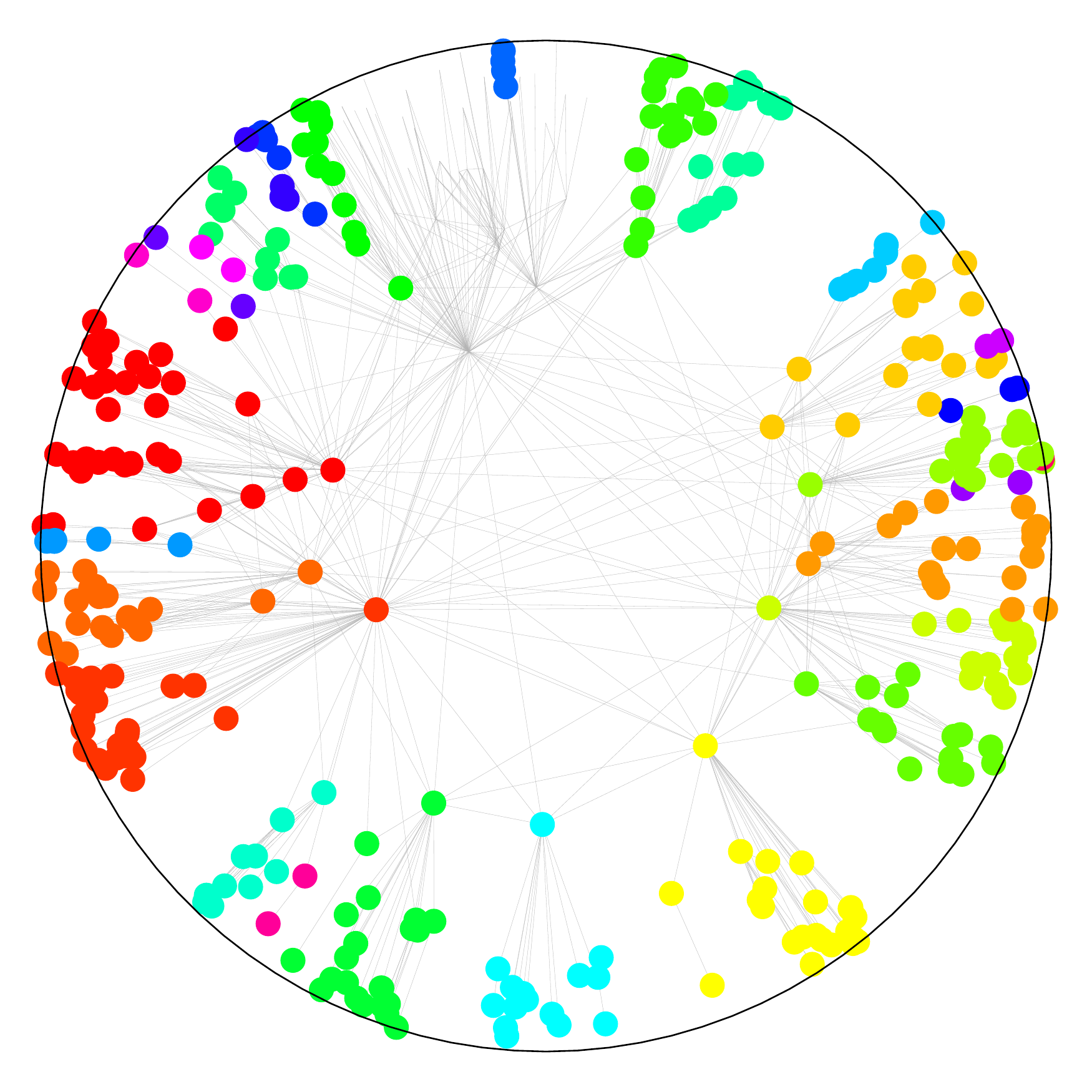}
        \end{minipage}
    }
    \subfigure[Modularity]{
        \begin{minipage}[t]{0.3\linewidth}
        \centering
        \includegraphics[width=1.1in]{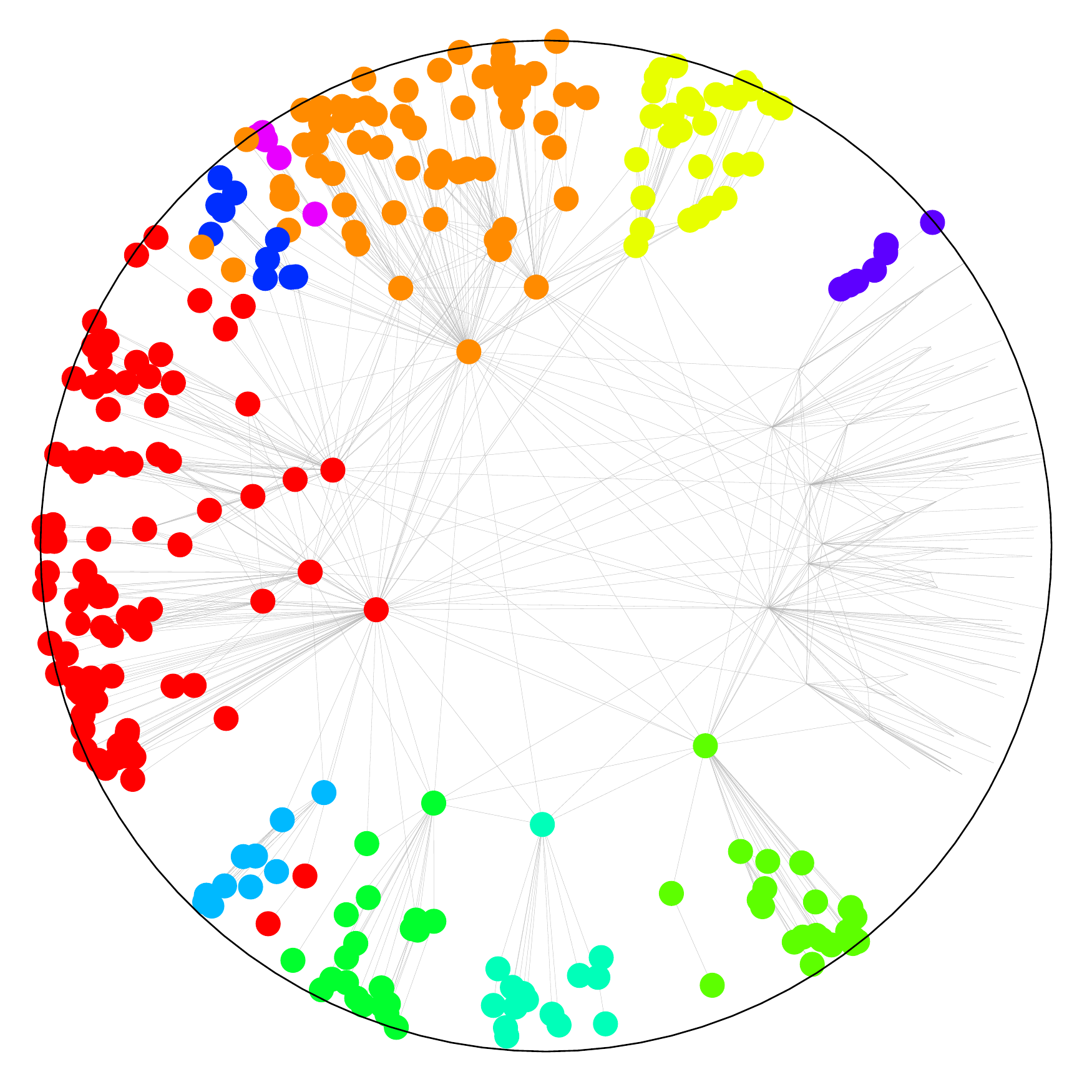}
        \end{minipage}
    }
    \subfigure[SBM-MDL]{
        \begin{minipage}[t]{0.3\linewidth}
        \centering
        \includegraphics[width=1.1in]{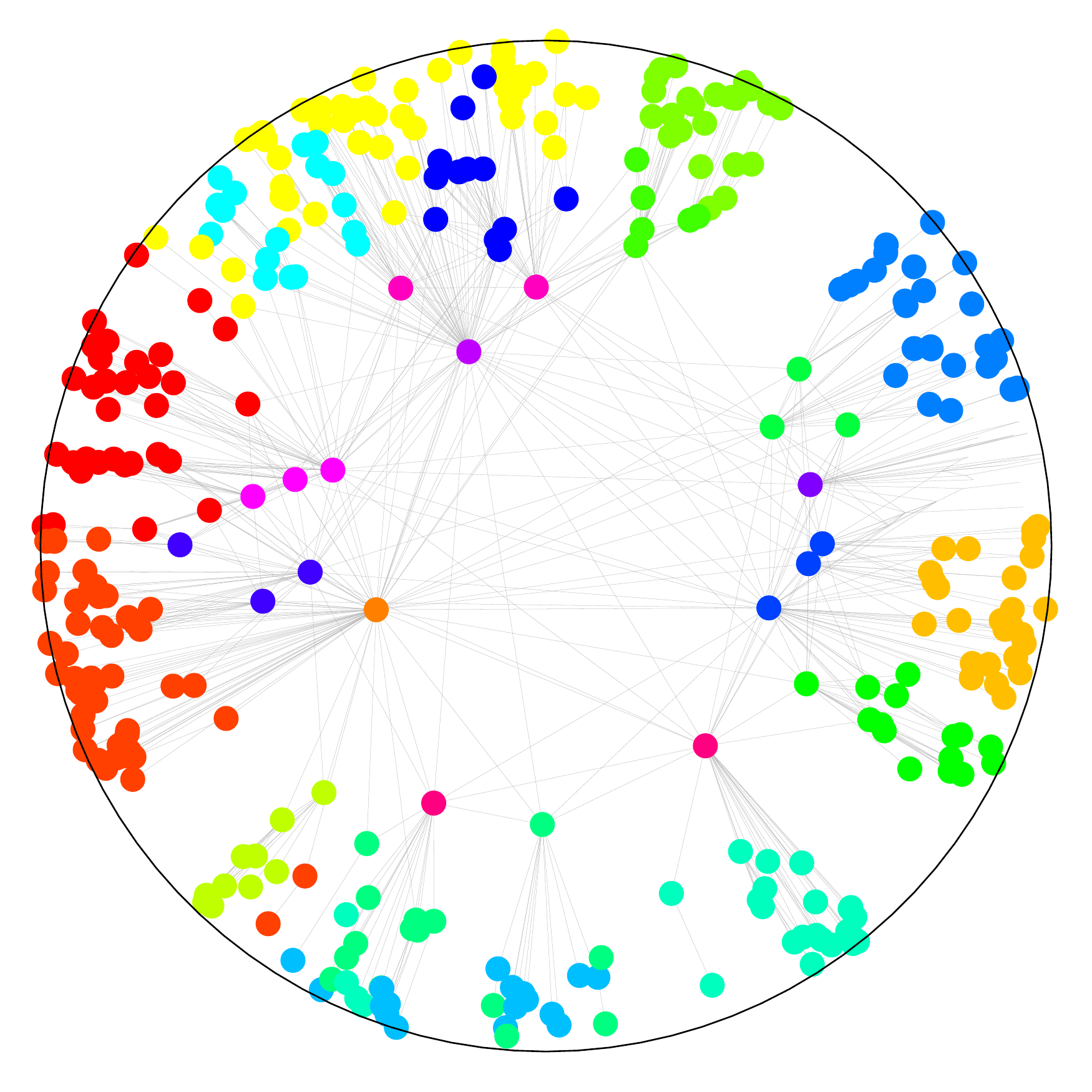}
        \end{minipage}
    }

    \subfigure[DCSBM-MDL]{
        \begin{minipage}[t]{0.3\linewidth}
        \centering
        \includegraphics[width=1.1in]{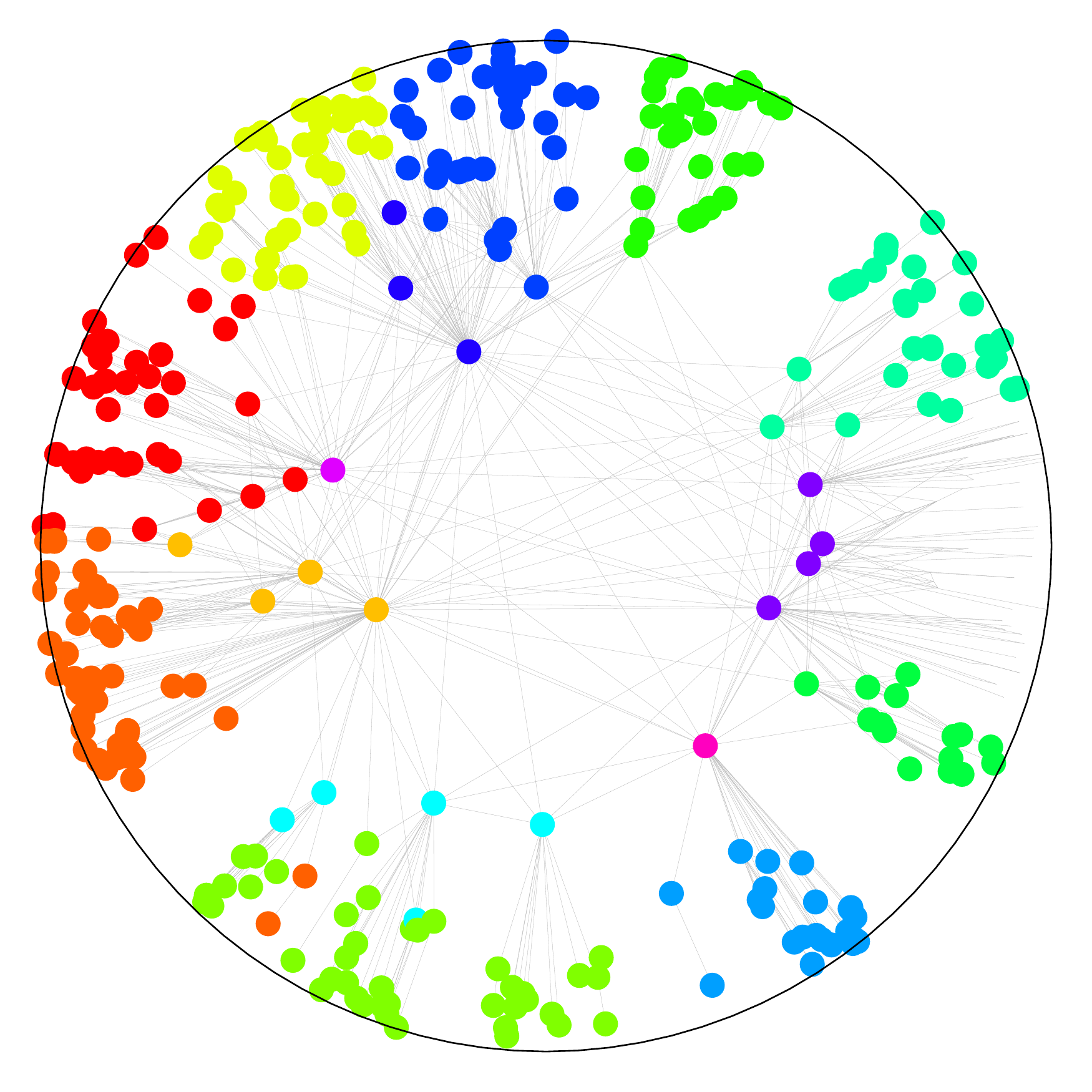}
        \end{minipage}
    }
    \subfigure[SBM]{
        \begin{minipage}[t]{0.3\linewidth}
        \centering
        \includegraphics[width=1.1in]{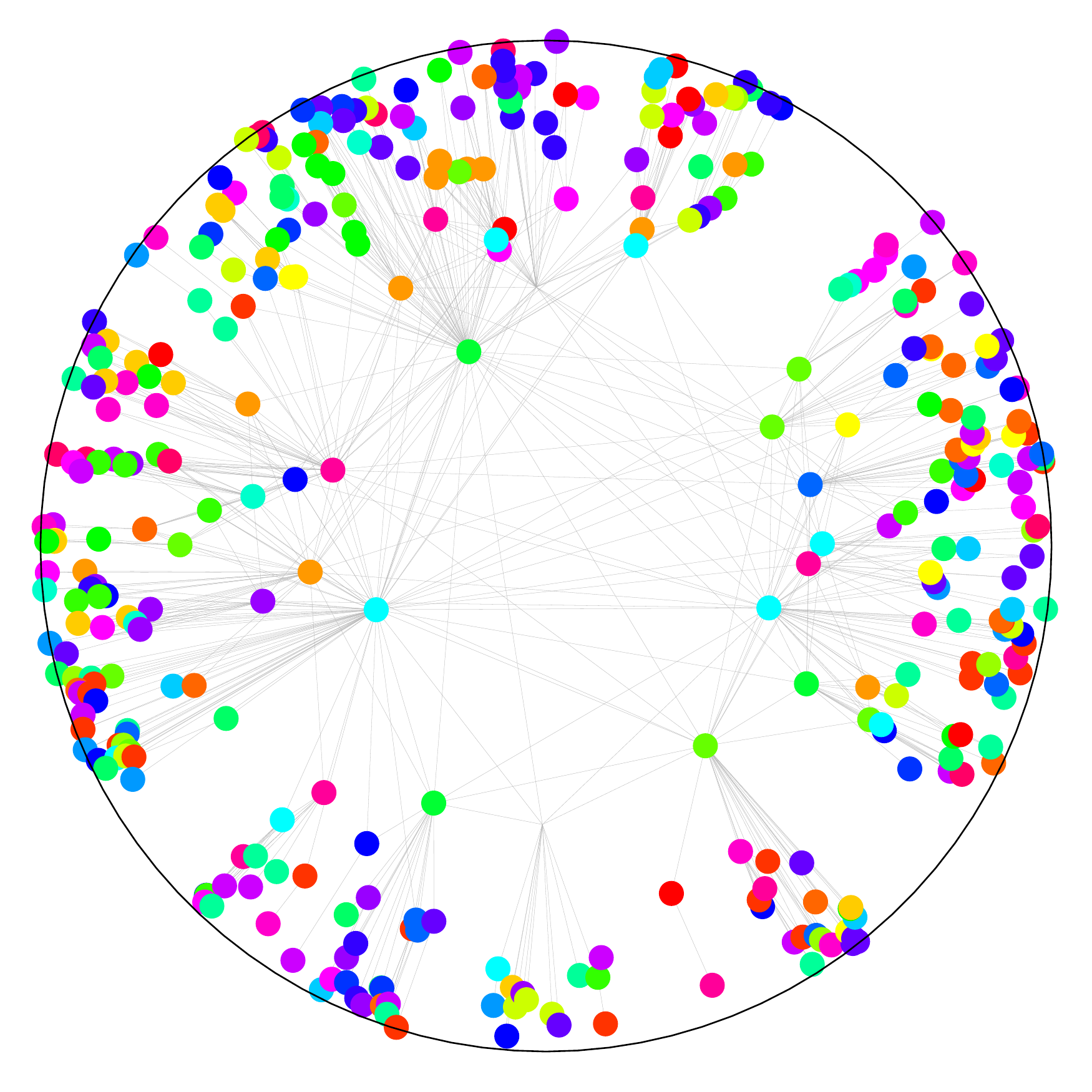}
        \end{minipage}
    }
    \subfigure[DCSBM]{
        \begin{minipage}[t]{0.3\linewidth}
        \centering
        \includegraphics[width=1.1in]{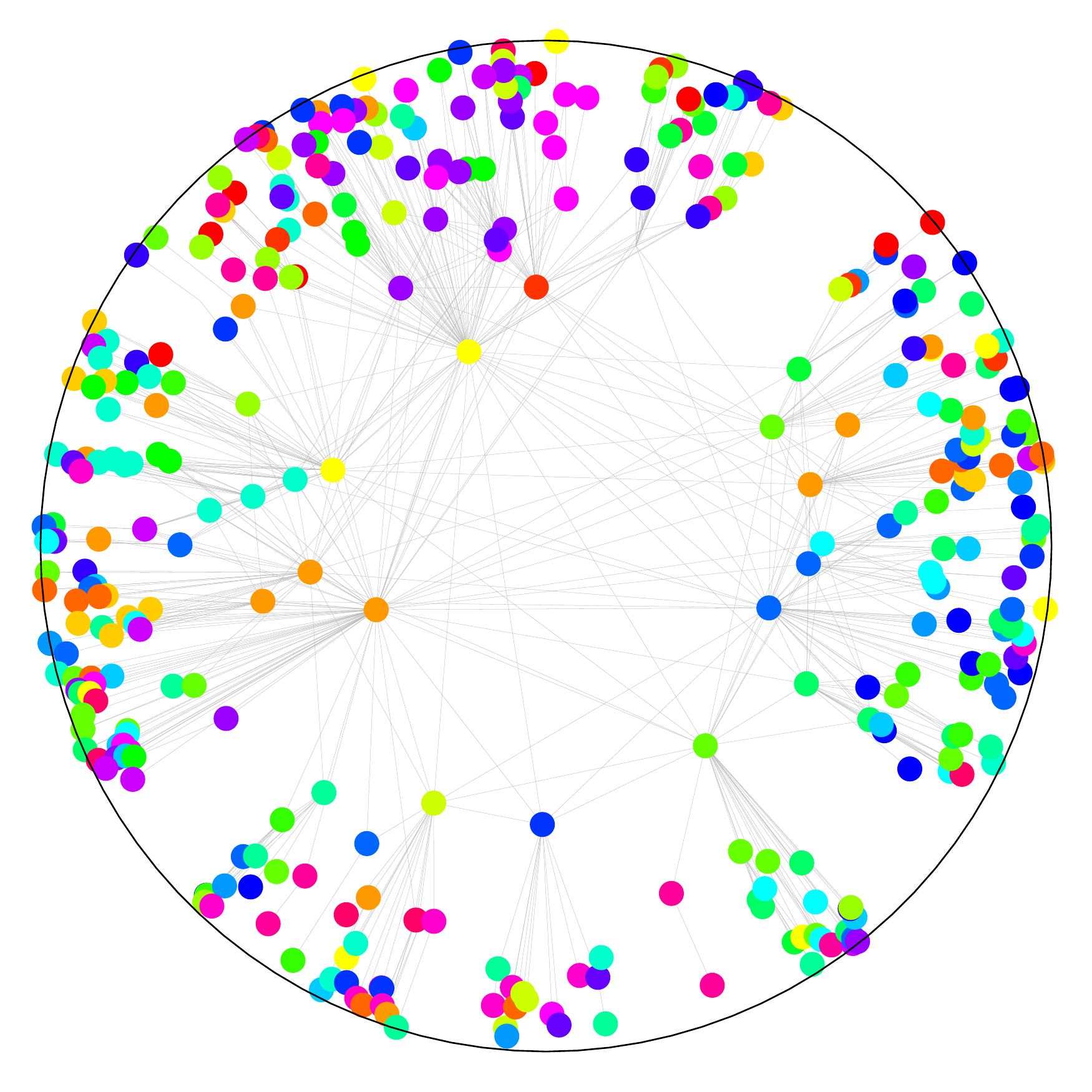}
        \end{minipage}
    }
    \centering
    \caption{Illustration for different clustering algorithm on the graph generated by random hyperbolic disk model. The nodes' color correspond to their module membership. The parameters are $C=2$, $\alpha=0.6$, $N=367$, $T=0.1$. For SBM and DCSBM, we set the block number to 30 for fair comparison which is the best partition number returned by infomap.}
    \label{ClusteringIllustration}
\end{figure}

\section{Multiplex Geometric Correlation}
Previous research has revealed the hidden correlation among radial and angular coordinates in multiplex network\cite{kleineberg2016hidden}: node pair located at small angular distance in one layer tends to be similar in another layer. We here show a new result which reveals a absolute correlation among node's multilayer angular coordinates instead of the relative one.  We denote $d_u^{{\rm max}}$, $d_u^{{\rm min}}$, for the maximal , minimal degree of node $u$ in each layer respectively. $\Delta\theta_u$ denotes the maximal absolute difference of angular coordinates of node $u$. 
Then we count the violation ratio of the following statistic:
\begin{equation}
\begin{split}
    &{\rm CN}(u)>{\rm CN}(v)\\
    &\cap\quad d_u^{{\rm max}} < d_v^{{\rm min}}\\
    &\cap\quad \Delta\theta_u < \Delta\theta_v
\end{split}
\end{equation}

Intuitively, this statistic assumes that if node $u$ has more common neighbors than $v$ while with lower degree then the angular distances among each layer should be narrower. Table~\ref{correlation} reports the violation ratio on multiple multiplex networks from 6 different regions. The result shows that the violation ratio is less than 5\% on all datasets which reveals the rationality of the assumption. It should be noted that our new result is different with\cite{kleineberg2016hidden} which considers the angular difference for node pairs among different layers. While our result reveals that the angular coordinates for a specific node is also similar given that this node has the similar local topology in each layer. 

This empirical finding inspires us to incorporate the absolute angular correlation for nodes in multiplex network, which is to align node's angular coordinates among layers with the similar topology, i.e., community structure.

\begin{table}[]
\begin{tabular}{|c|c|c|c|}
\hline
Dataset & \begin{tabular}[c]{@{}c@{}}$\sharp$ common \\ nodes\end{tabular} & violation ratio & \begin{tabular}[c]{@{}c@{}}violation ratio\\ (without degree)\end{tabular} \\ \hline
Router & 4819 & \begin{tabular}[c]{@{}c@{}}1.74\% \\ (227064/13069968)\end{tabular} & \begin{tabular}[c]{@{}c@{}}18.32\%\\  (4254648/23222761)\end{tabular} \\ \hline
Air Train & 69 & \begin{tabular}[c]{@{}c@{}}0.30\% \\ (8/2702)\end{tabular} & \begin{tabular}[c]{@{}c@{}}17.85\% \\ (850/4761)\end{tabular} \\ \hline
D.M. & 557 & \begin{tabular}[c]{@{}c@{}}3.28\% \\ (6082/185616)\end{tabular} & \begin{tabular}[c]{@{}c@{}}19.92\% \\ (61813/310249)\end{tabular} \\ \hline
C.E. & 238 & \begin{tabular}[c]{@{}c@{}}3.63\% \\ (1056/29124)\end{tabular} & \begin{tabular}[c]{@{}c@{}}12.08\% \\ (6842/56644)\end{tabular} \\ \hline
Brain & 77 & \begin{tabular}[c]{@{}c@{}}4.37\% \\ (122/2790)\end{tabular} & \begin{tabular}[c]{@{}c@{}}20.53\% \\ (1217/5929)\end{tabular} \\ \hline
arXiv & 1514 & \begin{tabular}[c]{@{}c@{}}2.76\% \\ (40748/1475242)\end{tabular} & \begin{tabular}[c]{@{}c@{}}22.20\% \\ (508888/2292196)\end{tabular} \\ \hline
\end{tabular}
\caption{Absolute angular correlation on multiplex networks from 6 different regions. We report for each dataset the violation ratio with and without degree consideration respectively.}
\label{correlation}
\end{table}

\section{Optimization Objectives}

Putting all above three ideas together, the loss function to optimize the hyperbolic embedding of multiplex network is:
\begin{equation}
\mathcal{L}(\Theta, H, R)=O_1(\Theta, R) + O_2(H, \Theta) + O_3(H, \Theta)
\label{total_objective}
\end{equation}

Then our final output is:
\begin{equation}
(\Theta^*, H^*, R^*)={\rm argmin}\mathcal{L}(\Theta, H, R)
\end{equation}

The intuitive idea is to pipeline the community detection and network embedding. However, as revealed by previous researches, these two tasks benefit each other. One intuitive reason is that real networks are always sparse and noisy. While network embedding always incorporates higher order proximity which improves the node representation performance. Then community detection based on network embedding will enhance the performance.

We will first outline our optimization process.
\begin{enumerate}
    \item Based on observed network topology, preliminary latent coordinates will be inferred.
    \item Multiplex community detection will then be performed through multiplex map equation. It should be noted that the random walker is guided by latent hyperbolic coordinates.
    \item Based on current multiplex community detection result, the angular coordinates within one community will be forcing closing.
\end{enumerate}

\section{Optimization}
\subsection{Inference Algorithm}
We optimize the objective function Eq.\ref{total_objective} of HME by running SGD over each component iteratively. HME first initialize the angular coordinates by the infomap community detection result. The radial coordinates are initialized by the analytical expression. Then we perform the RSGD optimization over the initialized hyperbolic coordinates. Augular coherence optimization is further performed to improve the community separation. Given the inferred hyperbolic coordinates, we can perform the infomap community detection to enable the feedback effect. Empirically, the iteration process converges quickly after several rounds.

\setlength{\textfloatsep}{0cm}
\setlength{\floatsep}{0cm}
\begin{algorithm}
    \caption{Inference algorithm for single layer HME}
    \label{algorithm1}
    \SetKwInOut{Input}{Input}
    \SetKwInOut{Output}{Output}
    \Input{
    \textit{$G$}: undirected unweighted graph $G=(V,E)$\\
    }
    \Output{
    \textit{$\Phi$}: node embedding $\Phi\in\mathcal{R}^{N*2}$. Each embedding has two components: radial and angular coordinate.\\
    \textit{$\Pi$}: community assignment.\\
    }
    Initialize $\Pi\leftarrow\textit{Infomap}(G))$\;
    Initialize node's radial coordinate ${r_i}$ by Eq.\ref{r_init}\;
    \While{$iter < T$}{
        initialize node angular coordinates by Eq. ?\;
        \For{\textbf{all} \textit{edges} $(i, j)\in E$}
        {
            RSGD on $\phi_i=(\theta_i, r_i)$ by Eq.\ref{rsgd_theta} and Eq.\ref{rsgd_r}\;
        }
        \For{\textbf{all} \textit{nodes} $i\in V$}
        {
            SGD on angular coordinate coherence by Eq.\ref{sgd_coherence}\;
        }
        Community Detection over poincare disk by Eq.?\;
    }
\end{algorithm}
\setlength{\textfloatsep}{0cm}
\setlength{\floatsep}{0cm}

\subsection{First order proximity}
Following \cite{nickel2017poincare}, we optimize Equation~\ref{O1} by employing the stochastic Riemannian optimization methods such as RSGD\cite{bonnabel2013stochastic}. However, we preform the optimization on the poincare disk model with polar coordinates due to its closely related analogy with the community structure for the embeddings. The general Riemannian optimization framework works as follows:

$$\boldsymbol{\theta}_{t+1}=\mathcal{R}_{\boldsymbol{\theta}_t}(-\eta_t\nabla_R\mathcal{L}(\boldsymbol{\theta}_t))$$

where $\boldsymbol{\theta}_t$ denotes the polar coordinate $(r_t, \rho_t)$ at iteration $t$ and $\mathcal{R}_{\boldsymbol{\theta}_t}$ denotes the retraction onto poincare disk at point $\boldsymbol{\theta}_t$ ($\mathcal{R}_{\boldsymbol{\theta}(v)}=\boldsymbol{\theta}+v$) which is actually a first order approximation of exponential map operation\cite{bonnabel2013stochastic}. $\eta_t$ denotes the learning rate. To derive the Riemannian gradient $\nabla_R$ for poincare disk with polar coordinates, we need to rescale the Euclidean gradient with the inverse of metric tensor matrix which is given below:
\begin{equation}
g_p=\left[
\begin{matrix}
 1 & 0\\
 0 & {\rm sinh^2}r\\
\end{matrix}
\right].
\label{metrictensor}
\end{equation}

The complete update procedure then follows as:
$$\boldsymbol{\theta}_{t+1}={\rm proj}\left(\boldsymbol{\theta}_t-\eta_tg_p^{-1}(\boldsymbol{\theta}_t)\nabla_E\mathcal{L}(\boldsymbol{\theta})\right).$$
where ${\rm proj}$ denotes the projection operation to pull the updated embedding coordinate back into the poincare disk.

The Euclidean gradient $\nabla_E$ is the product of the derivative of first order loss function with respect to poincare distance function and the partial derivative of the distance function, i.e.
$$\nabla_E=\frac{\partial\mathcal{L}}{\partial d(\boldsymbol{u}, \boldsymbol{v})}\frac{\partial d(\boldsymbol{u}, \boldsymbol{v})}{\partial \boldsymbol{u}}$$.

Recall that for node $v$ the local loss function is:
\begin{equation}
\mathcal{L}(v)=\sum\limits_{u\sim v}{\rm logP}(d(\boldsymbol{u},\boldsymbol{v}))+\sum\limits_{u\nsim v}{\rm log}(1-{\rm P}(d(\boldsymbol{u},\boldsymbol{v})))
\end{equation}
where the first term corresponds to the positive edges while the second models the negative edges.

Then the derivative of the local loss function $\mathcal{L}(v)$ with respect to distance function is:
$$\cdot\frac{\partial\mathcal{L}(v)}{\partial d(\boldsymbol{u}, \boldsymbol{v})}=\frac{1}{2T}\left(\sum\limits_{u\sim v}(1-\delta(x(\boldsymbol{u}, \boldsymbol{v})))+\sum\limits_{u\nsim v}(0-\delta(x(\boldsymbol{u},\boldsymbol{v})))\right)$$
where $\delta(\cdot)$ denotes the sigmoid function and $x(\boldsymbol{u},\boldsymbol{v})=\frac{R-d(\boldsymbol{u},\boldsymbol{v})}{2T}$.

To reduce the computation cost of local loss function which requires summation over the entire set of edges, following\cite{tang2015line} we employ the negative sampling approach to sample some fixed number of negative edges. The following function specifies the negative sampling loss
\begin{equation}
2T\cdot\frac{\partial\mathcal{L}(v)}{\partial d(\boldsymbol{u}, \boldsymbol{v})}=\sum\limits_{u\sim v}(1-\delta(x(\boldsymbol{u},\boldsymbol{v})))+\sum\limits_{i=1}^K \mathbb{E}_{u\sim P(v)}\left(0-\delta(x(\boldsymbol{u}, \boldsymbol{v}))\right),
\end{equation}
where $P(v)\propto d_v^{3/4}$ as proposed in previous work\cite{mikolov2013distributed}.

The partial derivative of the poincare distance function with respect to radial and angular coordinate are as following
$$\frac{\partial d(\boldsymbol{u},\boldsymbol{v})}{\partial r_u}=\frac{{\rm sinh}(r_u){\rm cosh}(r_v)-{\rm cosh}(r_u){\rm sinh}(r_v){\rm cos}(\rho_u-\rho_v)}{\sqrt{D^2-1}},$$
$$\frac{\partial d(\boldsymbol{u},\boldsymbol{v})}{\partial \rho_u}=\frac{{\rm sinh}(r_u){\rm sinh}(r_v){\rm sin}(\rho_u-\rho_v)}{\sqrt{D^2-1}}$$
where
$$D={\rm cosh}(r_u){\rm cosh}(r_v)-{\rm sinh}(r_u){\rm sinh}(r_v){\rm cos}(\rho_u-\rho_v).$$

Since the high quality estimate of radial coordinate has been derived in previous work\cite{bogu?á2010sustaining} as the following analytical expression:
\begin{equation}
r_u={\rm min}\left\{R, 2{\rm log}\left(\frac{2N\alpha T}{{\rm deg}(u) {\rm sin}(\pi T)(\alpha-\frac{1}{2})}\right)\right\}.
\label{r_init}
\end{equation}
We will initialize the radial coordinate as the estimated value and update the radial and angular coordinates as the following:
\begin{equation}
r_v^{(t+1)}=r_v^{(t)}-\eta_t\cdot\frac{\partial\mathcal{L}}{\partial d(\boldsymbol{u}, \boldsymbol{v})}\cdot\frac{\partial d(\boldsymbol{u}, \boldsymbol{v})}{\partial r_v}.
\label{rsgd_r}
\end{equation}
\begin{equation}
\rho_v^{(t+1)}=\rho_v^{(t)}-\eta_t\cdot{\rm sinh}^{-2}(r_v)\cdot\frac{\partial\mathcal{L}}{\partial d(\boldsymbol{u}, \boldsymbol{v})}\cdot\frac{\partial d(\boldsymbol{u}, \boldsymbol{v})}{\partial \rho_v}.
\label{rsgd_theta}
\end{equation}

\subsection{Angular coordinate coherence}

We optimize Equation \ref{O3} by stochastic gradient descent. For each node $u$ and current partition we compute the gradient of angular coordinate as following
\begin{equation}
\frac{\partial O_3}{\partial \theta_u} = \frac{(-sin\theta_u)\cdot\sum\limits_{v\in C_m}cos\theta_v+cos\theta_u\cdot\sum\limits_{v\in C_m}sin\theta_v}{\sqrt{\left(\sum\limits_{v\in C_m}cos\theta_v\right)^2+\left(\sum\limits_{v\in C_m}sin\theta_v\right)^2}},
\label{sgd_coherence}
\end{equation}
where
\begin{equation}
\delta_u=m.
\end{equation}

\section{Experiments}

\subsection{Datasets}\label{datasetSection}
The open multiplex network dataset released by Kaj-Kolja Kleineberg is used in our experiments. The dataset consists of different multiplex networks from diverse domains. They were used to reveal the hidden geometric correlation in multiplex systems\cite{kleineberg2016hidden}\cite{kleineberg2017geometric} and characterize the analogy between hyperbolic embedding and community structure\cite{faqeeh2018characterizing}. We consider 6 multiplex network dataset in this paper. The statistics of datasets are listed in Table~\ref{ResolutionLimit}. The details of the dataset are as follows:

\textbf{Routers} is a two layer network of Autonomous Systems (ASs) collected by CAIDA. Each node is a AS. And the link between ASs denotes the traffic exchange. The two layer represents IPv4 and IPv6 packet routing between ASs.

\textbf{C. Elegans Connectomme} is a three layer network. In this dataset, the nodes are neurons. And the three layers correspond to different type of synaptic connection (electric,
chemical monadic, and chemical Polyadic).

\textbf{Drosophila Melanogaster} network represents the interactions between proteins. The two layers correspond to suppressive genetic interaction and additive genetic interaction respectively.

\textbf{arXiv} dataset consists of 6 layer networks. The nodes are authors. The edge indicates a coauthorship between two authors. Each layer correspond to one category.

\textbf{SacchPomb} dataset consists of 2 layer networks and represents the multiplex genetic and protein interaction network of Saccharomyces Pombe. Each layer corresponds to different types of interactions (direct interaction and colocalization).

\textbf{Rattus} dataset represents the multiplex genetic and protein interaction network of the Rattus Norvegicus and consists two layer of networks.

\begin{figure*}[htbp]
    \centering
    \subfigure[iteration=0]{
        \begin{minipage}[t]{0.18\linewidth}
        \centering
        \includegraphics[width=1.2in]{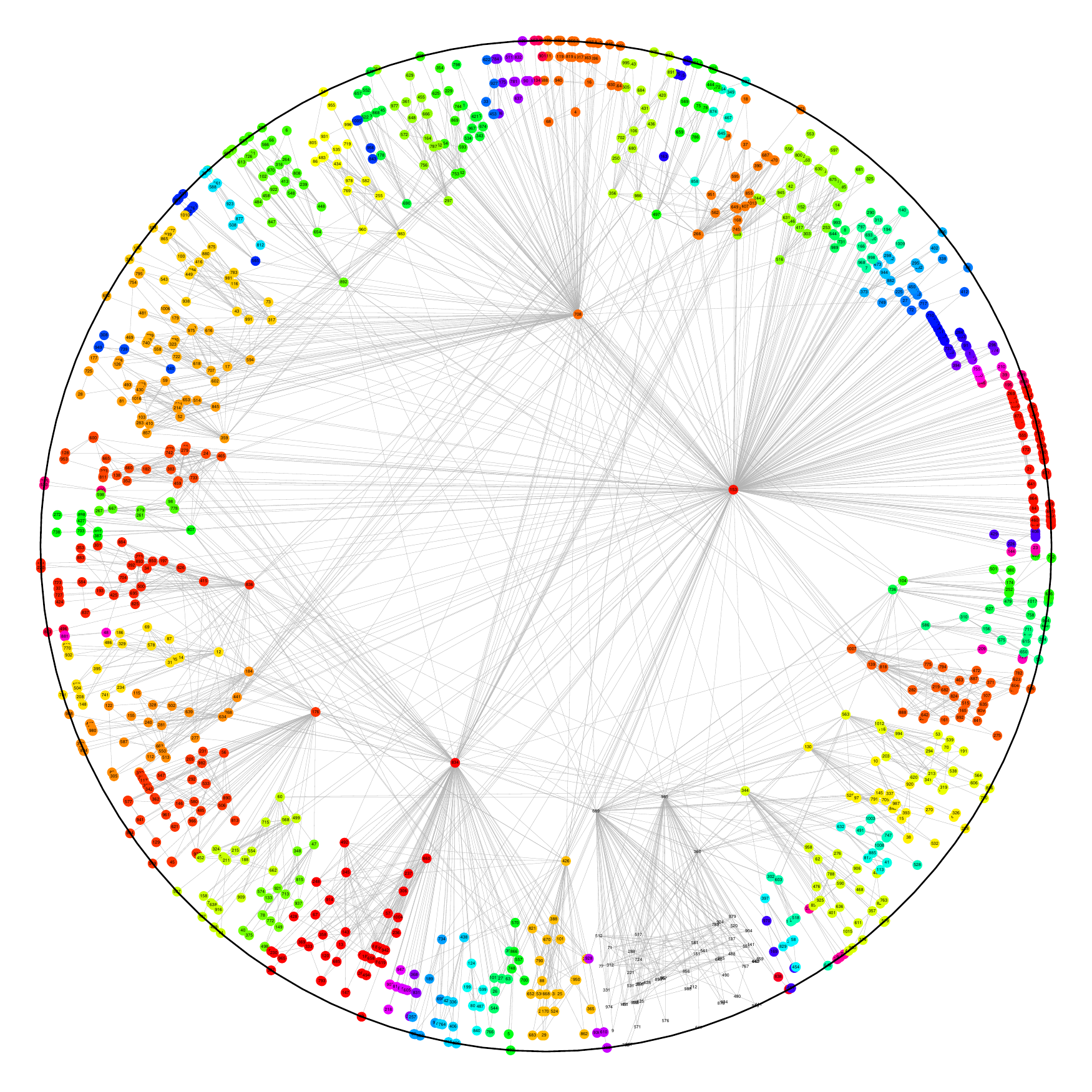}
        \end{minipage}
    }
    \subfigure[iteration=25]{
        \begin{minipage}[t]{0.18\linewidth}
        \centering
        \includegraphics[width=1.2in]{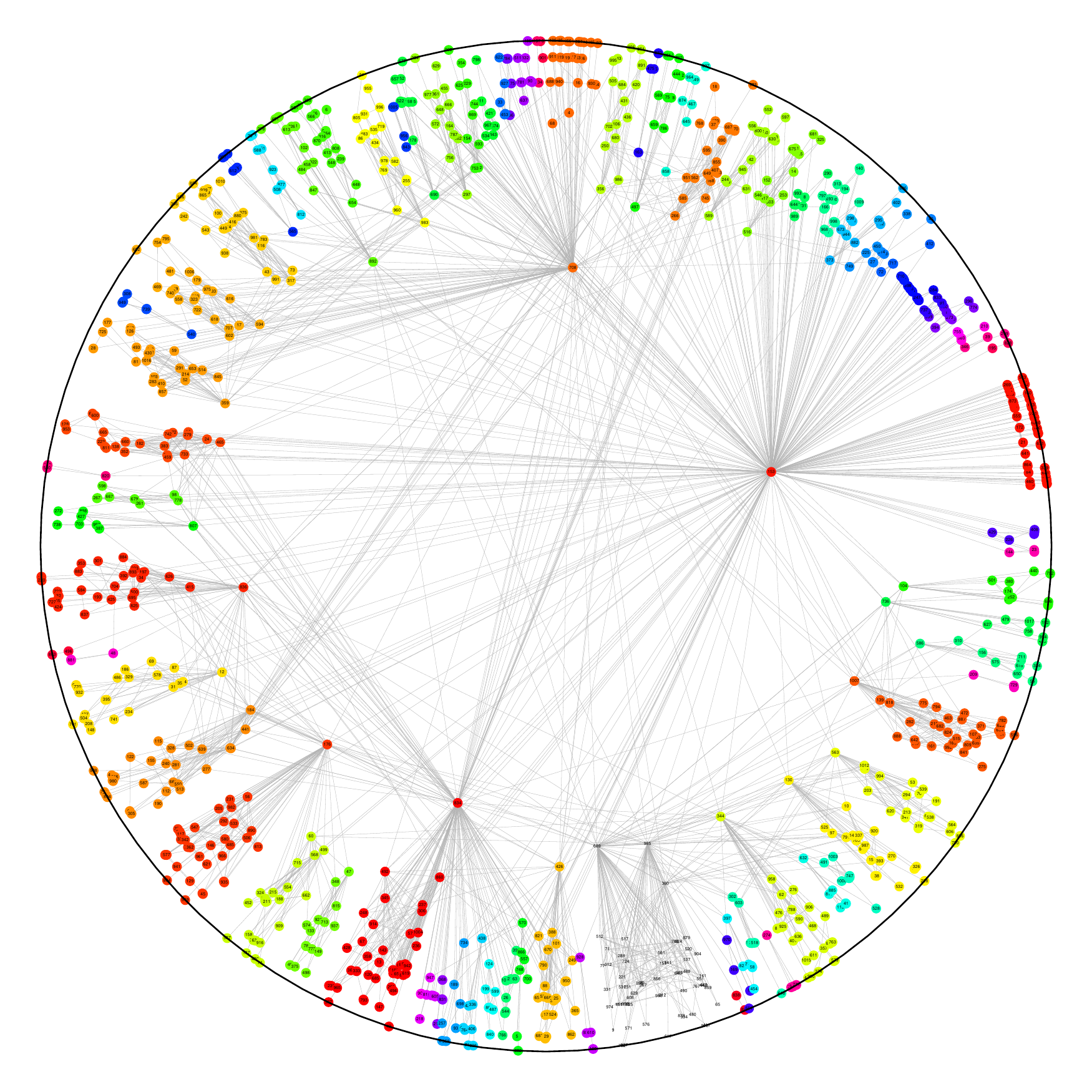}
        \end{minipage}
    }
    \subfigure[iteration=50]{
        \begin{minipage}[t]{0.18\linewidth}
        \centering
        \includegraphics[width=1.2in]{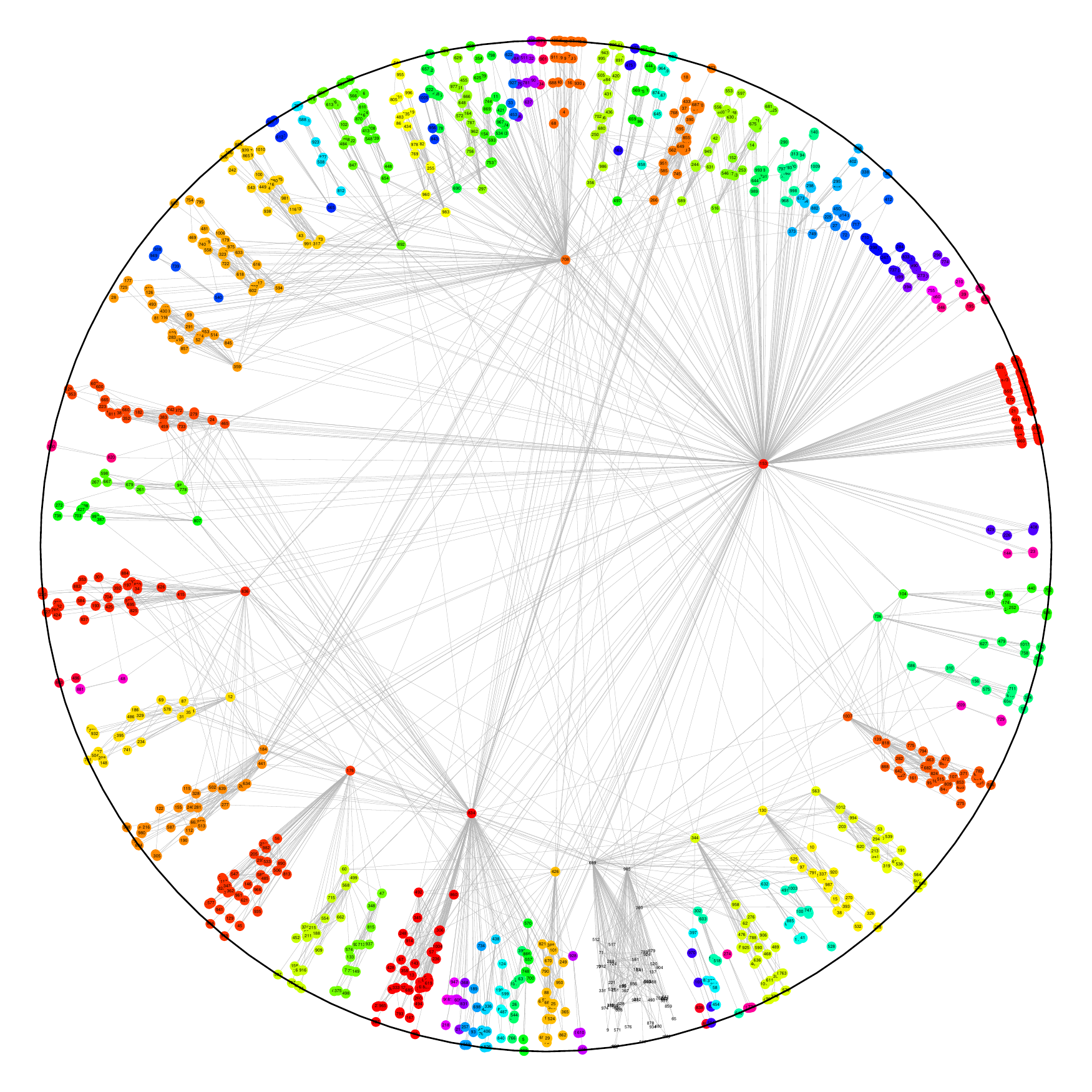}
        \end{minipage}
    }    
    \subfigure[iteration=75]{
        \begin{minipage}[t]{0.18\linewidth}
        \centering
        \includegraphics[width=1.2in]{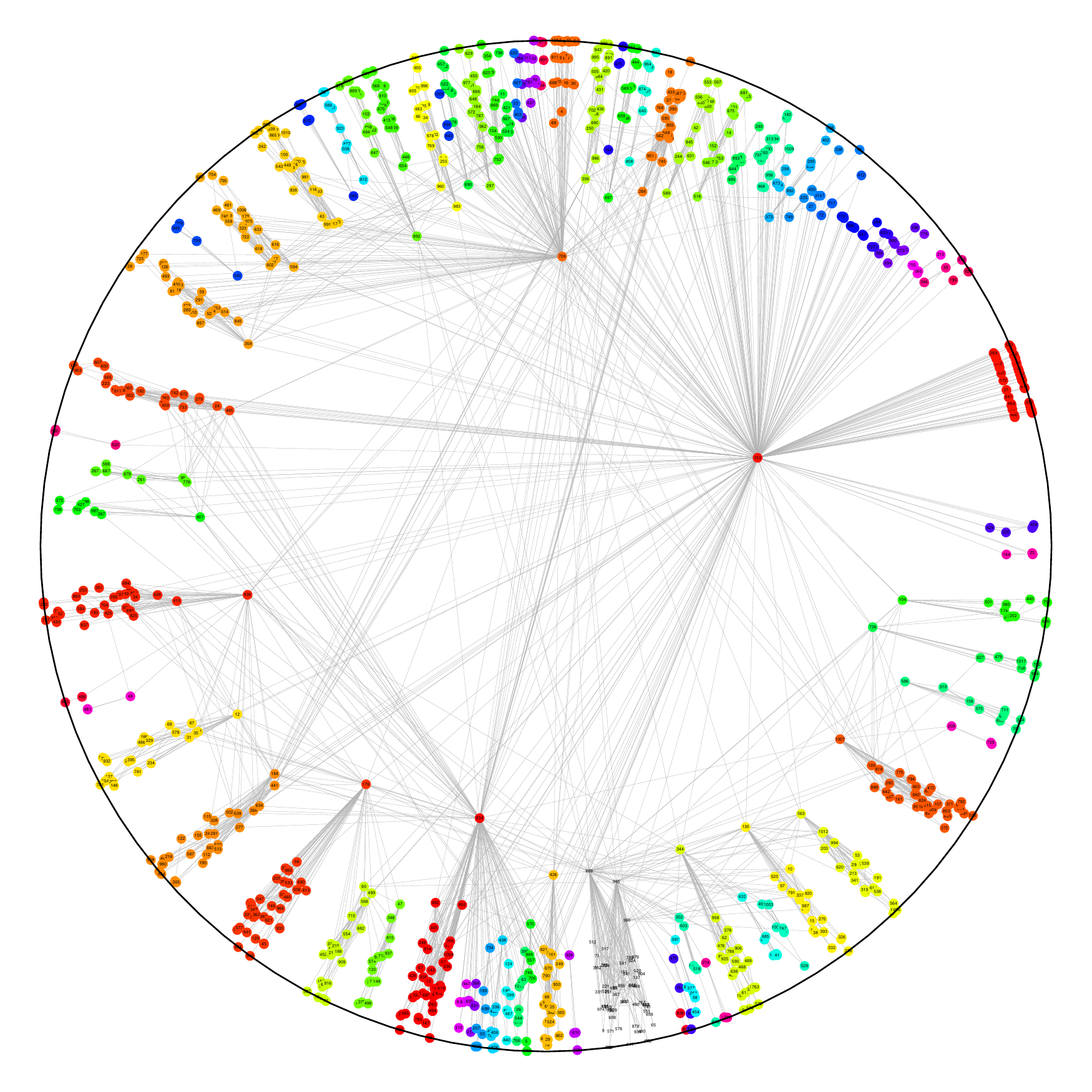}
        \end{minipage}
    }
    \subfigure[iteration=100]{
        \begin{minipage}[t]{0.18\linewidth}
        \centering
        \includegraphics[width=1.2in]{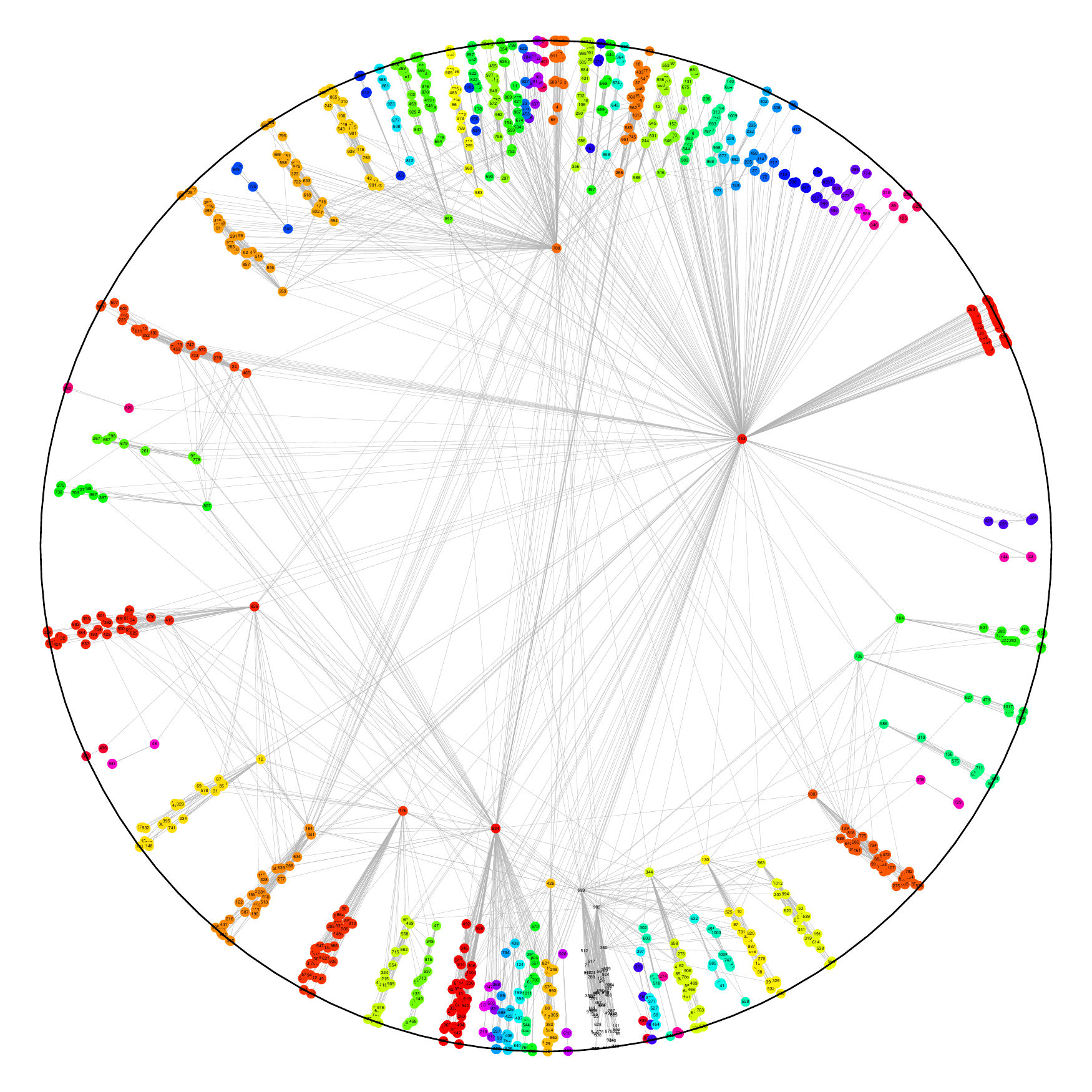}
        \end{minipage}
    }    
    \caption{poincare disk embedding during the 100 iteration training process}
    \label{training}
\end{figure*}

\subsection{Baselines}

We consider the following algorithms as our baselines. For comparison and clearance, we list and classify the baseline algorithms in Table~\ref{BaselineClassification}.

\begin{table}[]
\small
\begin{tabular}{|c|c|c|c|c|c|}
\hline
                                                                & \begin{tabular}[c]{@{}c@{}}node \\ embed\end{tabular} & \begin{tabular}[c]{@{}c@{}}com. \\ embed\end{tabular} & \begin{tabular}[c]{@{}c@{}}com.\\  detection\end{tabular} & \begin{tabular}[c]{@{}c@{}}embed \\ space\end{tabular} & \begin{tabular}[c]{@{}c@{}}network \\ layer\end{tabular} \\ \hline
LINE                      & $\checkmark$                                             &                                                       &                                                           & E                                                      & S                                                        \\ \hline
ComE             & $\checkmark$                                             & $\checkmark$                                             & $\checkmark$                                                 & E                                                      & S                                                        \\ \hline
MNE                  & $\checkmark$                                             &                                                       &                                                           & E                                                      & M                                                        \\ \hline
Infomap                &                                                       &                                                       & $\checkmark$                                                 & -                                                      & S                                                        \\ \hline
HyperMap & $\checkmark$                                             &                                                       &                                                           & H                                                      & S                                                        \\ \hline
GMM               &                                                       &                                                       & $\checkmark$                                                 & H                                                      & M                                                        \\ \hline
CHM                 & $\checkmark$                                             &                                                       &                                                           & H                                                      & S                                                        \\ \hline
SEAM                      &                                                       &                                                       & $\checkmark$                                                 & H                                                      & S                                                        \\ \hline
SpringEmb    & $\checkmark$                                             &                                                       &                                                           & H                                                      & S                                                        \\ \hline
Poincare            & $\checkmark$                                             &                                                       &                                                           & H                                                      & S                                                        \\ \hline
Coalescent        & $\checkmark$                                             &                                                       & $\checkmark$                                                          & H                                                      & S                                                        \\ \hline
\end{tabular}
\caption{Details and classification of baseline algorithms. 'E' and 'H' in column embed space stand for Euclidean and Hyperbolic space respectively. 'S' and 'M' in column network layer stand for single layer and multiplex respectively.}
\label{BaselineClassification}
\end{table}

\textbf{LINE\cite{tang2015line}.} LINE is the one of the most influential euclidean embedding algorithms. This algorithm assumes each node is with a euclidean vector. And two node connect with each other in the probability of inner product of their embeddings. LINE employs the first order and second order topology proximity to estimate the latent euclidean embedding.

\textbf{ComE\cite{cavallari2017learning}.} ComE proposes an algorithm jointly learning the community embedding and node embedding. ComE assumes the node embeddings within one community follow the Gaussian Mixture Model. Then besides the first order and second order proximity, the community coherence is also introduced to enfore the nodes embedding close to each other in one community.

\textbf{MNE\cite{zhang2018scalable}.} MNE is an embedding model for multiplex network. MNE assumes each node's embedding is an linear combination of two parts. The common vector part is shared across different layers. While layer specific vector is distinct for each type of relationship.

\textbf{Hypermap\cite{papadopoulos2015networkcommon}.} Hypermap is an hyperbolic embedding algorithm for single layer network. Hypermap works by optimizing the likelihood of the combination of first and second order link probability.

\textbf{GMM\cite{kleineberg2016hidden}.} Geometric Multiplex Model (GMM) is proposed to model the multiplex hyperbolic network model. This work does not directly compute the node's embedding. While a straightforward multiplex community detection algorithm is proposed in this work. We include this method as a baseline to evaluate our algorithm's multiplex community detection performance.

\textbf{CHM\cite{wang2016hyperbolic}.} CHM is the first work to show that community structure can be used to improve the accuracy and efficiency of hyperbolic embedding. CHM is a single layer oriented hyperbolic embedding algorithm.

\textbf{SEAM\cite{wang2016fast}.} SEAM is a community detection algorithm which employs the node's latent hyperbolic embeddings. We introduce this algorithm as one baseline method to evaluate our algorithm's community detection performance.

\textbf{PoincareEmbedding\cite{nickel2017poincare}.} $Poincar\acute{e}$Embedding is an efficient algorithm to learn the $poincar\acute{e}$ embedding of nodes based on Riemannian optimization. This algorithm is based on the $Poincar\acute{e}$ ball model as it is well suited for gradient optimization.

\textbf{SpringEmbedder\cite{blasius2018efficient}.} This algorithm improves the Hypermap by deriving an analytic expression which relates the common neighbor number and their angular difference.

\textbf{Infomap\cite{rosvall2008maps}.} Infomap is a nonparametric community detection algorithm. We include this method to evaluate our method's community detection performance.

\textbf{Coalescent embedding\cite{muscoloni2017machine}.} It adopts topology based machine learning for nonlinear dimension reduction to approximate the node angular organizations in the hyperbolic disk.

\subsection{Evaluation on Synthetic Networks}
To evaluate the accuracy of HME embedding algorithm, we first generate synthetic networks according to the hyperbolic random graph model. Specifically we employ the hyperbolic random graph sampling algorithm\footnote{https://hpi.de/friedrich/research/hyperbolic} proposed in \cite{BringmannKL19Geometric} which runs in expected linear time. The coordinates can then be considered as the ground truth embedding. However even for the synthetic networks, several different coordinates up to rotations and translations will be all equivalent. Then following \cite{muscoloni2017machine}, we employ the Pearson correlation between all the pairwise hyperbolic distances of nodes (denoted as HD-correlation) to evaluate the accuracy of embedding performance. We generate four different scale networks by combining varying hyperbolic random graph model parameters. For each parameter combination, we generate different synthetic random networks and compute the HD-correlation by average. The results are report in Table \ref{HDcorrelation}.

\definecolor{mygray}{gray}{.8}

\begin{table*}[]
\normalsize
\begin{threeparttable}          
\begin{tabular}{|llrrrrlrrrrlrrrr|}
\hline
                                             &                                 & \multicolumn{4}{c}{C=2}                                                                   &                      & \multicolumn{4}{c}{C=0}                                                                   &                      & \multicolumn{4}{c|}{C=-2}                                                                  \\ \cline{1-6} \cline{8-11} \cline{13-16}
\multicolumn{1}{|l|}{Size}                      & Method                          & T=0.1                & T=0.3                & T=0.6                & T=0.9                & \multicolumn{1}{r}{} & T=0.1                & T=0.3                & T=0.6                & T=0.9                & \multicolumn{1}{r}{} & T=0.1                & T=0.3                & T=0.6                & T=0.9                 \\ \cline{1-6} \cline{8-11} \cline{13-16}
\multicolumn{1}{|l|}{\multirow{9}{*}{N=100}}   & \multicolumn{1}{l|}{Degree}     & 3.13                 & 3.27                 & 3.78                 & 3.57                 &                      & 4.50                 & 4.45                 & 4.15                 & 8.22                 &                      & 9.15                 & 9.91                 & 10.35                & 10.55                 \\ \cline{2-16}
\multicolumn{1}{|l|}{}                       & \multicolumn{1}{l|}{Clustering} & 0.40                 & 0.37                 & 0.49                 & 0.29                 &                      & 0.64                 & 0.61                 & 0.38                 & 0.47                 &                      & 0.76                 & 0.65                 & 0.46                 & 0.47                  \\ \cline{2-16}
\multicolumn{1}{|l|}{}                       & \multicolumn{1}{l|}{Embedder}   & 0.83                 & \cellcolor{mygray}0.74        & 0.68                 & 0.64                 &                      & 0.72                 & 0.83                 & 0.79                 & 0.85                 &                      & 0.86                 & 0.92                 & 0.84                 & 0.83                  \\
\multicolumn{1}{|l|}{}                       & \multicolumn{1}{l|}{Spring}     & \cellcolor{mygray}{0.86}        & 0.73                 & 0.64                 & 0.70                 &                      & \cellcolor{mygray}{0.77}        & \cellcolor{mygray}{0.89}        & 0.75                 & \cellcolor{mygray}{0.82}        &                      & \cellcolor{mygray}{0.94}        & \cellcolor{mygray}{0.94}        & \cellcolor{mygray}{0.85}        & 0.84                  \\
\multicolumn{1}{|l|}{}                       & \multicolumn{1}{l|}{HyperMapCN} & 0.07                 & 0.06                 & \cellcolor{mygray}{0.73}        & 0.60                 &                      & 0.20                 & 0.35                 & 0.26                 & 0.42                 &                      & 0.50                 & 0.67                 & 0.64                 & 0.52                  \\
\multicolumn{1}{|l|}{}                       & \multicolumn{1}{l|}{Coalescent} & 0.63                 & 0.66                 & 0.69                 & \cellcolor{mygray}{0.72}        &                      & 0.67                 & 0.85                 & \cellcolor{mygray}{0.84}        & 0.82                 &                      & 0.90                 & 0.91                 & 0.83                 & \cellcolor{mygray}{0.86}         \\
\multicolumn{1}{|l|}{}                       & \multicolumn{1}{l|}{Poincare}   & 0.55                 & 0.51                 & 0.22                 & 0.43                 & \multicolumn{1}{r}{} & 0.48                 & 0.39                 & 0.35                 & 0.39                 & \multicolumn{1}{r}{} & 0.37                 & 0.36                 & 0.49                 & 0.56                  \\
\multicolumn{1}{|l|}{}                       & \multicolumn{1}{l|}{CHM}        & 0.71                 & 0.70                     & 0.63                 & 0.63                 & \multicolumn{1}{r}{} & 0.62                 & 0.62                 & 0.67                 & 0.73                 & \multicolumn{1}{r}{} & 0.74                 & 0.81                 & 0.79                 & 0.76                  \\
\multicolumn{1}{|l|}{}                       & \multicolumn{1}{l|}{HME}        & \textbf{0.87} & \textbf{0.86} & \textbf{0.75} & \textbf{0.73} &                      & \textbf{0.78} & \textbf{0.91} & \textbf{0.92} & \textbf{0.93} &                      & 0.93 & \textbf{0.94} & \textbf{0.86} & 0.83 \\ \hline
\multicolumn{1}{|l|}{\multirow{9}{*}{N=500}}   & \multicolumn{1}{l|}{Degree}     & 3.15                 & 3.18                 & 3.97                 & 4.39                 &                      & 5.19                 & 5.89                 & 5.25                 & 7.29                 &                      & 11.41                & 12.10                & 14.93                & 17.09                 \\ \cline{2-6} \cline{8-11} \cline{13-16}
\multicolumn{1}{|l|}{}                       & \multicolumn{1}{l|}{Clustering} & 0.51                 & 0.48                 & 0.43                 & 0.30                 &                      & 0.63                 & 0.61                 & 0.46                 & 0.42                 &                      & 0.77                 & 0.68                 & 0.53                 & 0.48                  \\ \cline{2-6} \cline{8-11} \cline{13-16}
\multicolumn{1}{|l|}{}                       & \multicolumn{1}{l|}{Embedder}   & 0.36                 & 0.37                 & 0.57                 & \cellcolor{mygray}{0.71}        &                      & 0.81                 & 0.78                 & 0.82                 & 0.82                 &                      & 0.86                 & 0.93                 & 0.92                 & 0.94                  \\
\multicolumn{1}{|l|}{}                       & \multicolumn{1}{l|}{Spring}     & \cellcolor{mygray}{0.43}        & \cellcolor{mygray}{0.44}        & \cellcolor{mygray}{0.64}        & 0.68                 &                      & \cellcolor{mygray}{0.86}        & \cellcolor{mygray}{0.85}        & 0.78                 & 0.87                 &                      & 0.88                 & 0.96                 & 0.95                 & 0.93                  \\
\multicolumn{1}{|l|}{}                       & \multicolumn{1}{l|}{HyperMapCN} & 0.34                 & 0.33                 & 0.03                 & 0.58                 &                      & 0.73                 & 0.74                 & 0.77                 & 0.73                 &                      & 0.81                 & 0.88                 & 0.92                 & 0.80                  \\
\multicolumn{1}{|l|}{}                       & \multicolumn{1}{l|}{Coalescent} & 0.31                 & 0.31                 & 0.43                 & 0.68                 &                      & 0.79                 & 0.72                 & \cellcolor{mygray}{0.84}        & \cellcolor{mygray}{0.91}        &                      & \cellcolor{mygray}{0.94}        & \cellcolor{mygray}{0.95}        & \cellcolor{mygray}{0.95}        & \cellcolor{mygray}{0.95}         \\
\multicolumn{1}{|l|}{}                       & \multicolumn{1}{l|}{Poincare}   & 0.01                 & 0.09                 & 0.34                 & 0.12                 &                      & 0.18                 & 0.08                 & 0.18                 & 0.20                 &                      & 0.17                 & 0.20                 & 0.18                 & 0.25                  \\
\multicolumn{1}{|l|}{}                       & \multicolumn{1}{l|}{CHM}        & 0.38                 & 0.40                 & 0.50                 & 0.58                 & \multicolumn{1}{r}{} & 0.56                 & 0.57                 & 0.60                 & 0.63                 & \multicolumn{1}{r}{} & 0.61                 & 0.66                 & 0.74                 & 0.73                  \\
\multicolumn{1}{|l|}{}                       & \multicolumn{1}{l|}{HME}        & \textbf{0.46} & \textbf{0.47} & \textbf{0.65} & 0.70 &                      & \textbf{0.86} & \textbf{0.87} & \textbf{0.87} & 0.89 &                      & 0.93 & 0.94 & 0.92 & 0.90 \\ \hline
\multicolumn{1}{|l|}{\multirow{9}{*}{N=1000}}  & \multicolumn{1}{l|}{Degree}     & 3.66                 & 3.69                 & 3.73                 & 4.29                 &                      & 5.54                 & 4.98                 & 5.74                 & 7.06                 &                      & 13.67                & 11.94                & 14.85                & 15.38                 \\ \cline{2-6} \cline{8-11} \cline{13-16}
\multicolumn{1}{|l|}{}                       & \multicolumn{1}{l|}{Clustering} & 0.58                 & 0.52                 & 0.47                 & 0.34                 &                      & 0.68                 & 0.59                 & 0.47                 & 0.42                 &                      & 0.79                 & 0.66                 & 0.52                 & 0.45                  \\ \cline{2-6} \cline{8-11} \cline{13-16}
\multicolumn{1}{|l|}{}                       & \multicolumn{1}{l|}{Embedder}   & 0.40                 & 0.45                 & 0.43                 & \cellcolor{mygray}{0.81}        &                      & \cellcolor{mygray}{0.87}        & 0.67                 & 0.78                 & \cellcolor{mygray}{0.86}        &                      & 0.94                 & 0.88                 & 0.95                 & 0.94                  \\
\multicolumn{1}{|l|}{}                       & \multicolumn{1}{l|}{Spring}     & \cellcolor{mygray}{0.45}        & \cellcolor{mygray}{0.49}        & \cellcolor{mygray}{0.51}        & 0.75                 &                      & 0.78                 & 0.72                 & 0.72                 & 0.74                 &                      & 0.98                 & 0.88                 & 0.96                 & 0.94                  \\
\multicolumn{1}{|l|}{}                       & \multicolumn{1}{l|}{HyperMapCN} & 0.39                 & 0.44                 & 0.00                 & 0.62                 &                      & 0.67                 & 0.62                 & 0.73                 & 0.74                 &                      & 0.81                 & 0.76                 & 0.84                 & 0.77                  \\
\multicolumn{1}{|l|}{}                       & \multicolumn{1}{l|}{Coalescent} & 0.37                 & 0.40                 & 0.36                 & 0.73                 &                      & 0.79                 & \cellcolor{mygray}{0.75}        & \cellcolor{mygray}{0.81}        & 0.75                 &                      & \cellcolor{mygray}{0.95}        & \cellcolor{mygray}{0.95}        & \cellcolor{mygray}{0.96}        & \cellcolor{mygray}{0.95}         \\
\multicolumn{1}{|l|}{}                       & \multicolumn{1}{l|}{Poincare}   & -0.02                & 0.00                 & 0.22                 & 0.20                 &                      & 0.15                 & 0.14                 & 0.14                 & 0.17                 &                      & 0.17                 & 0.20                 & 0.24                 & 0.20                  \\
\multicolumn{1}{|l|}{}                       & \multicolumn{1}{l|}{CHM}        & 0.39                 & 0.42                 & 0.42                 & 0.54                 & \multicolumn{1}{r}{} & 0.55                 & 0.50                 & 0.55                 & 0.61                 & \multicolumn{1}{r}{} & 0.59                 & 0.58                 & 0.70                 & 0.70                  \\
\multicolumn{1}{|l|}{}                       & \multicolumn{1}{l|}{HME}        & \textbf{0.50} & \textbf{0.52} & \textbf{0.52} & \textbf{0.83} &                      & 0.83 & \textbf{0.80} & \textbf{0.82} & \textbf{0.87} &                      & \textbf{0.95} & 0.93 & \textbf{0.96} & \textbf{0.95} \\ \hline
\multicolumn{1}{|l|}{\multirow{9}{*}{N=10000}} & \multicolumn{1}{l|}{Degree}     & 5.22                 & 4.98                 & 5.43                 & 4.68                 &                      & 6.09                 & 7.01                 & 6.60                 & 7.78                 &                      & 15.47                & 13.92                & 16.29                & 20.12                 \\ \cline{2-6} \cline{8-11} \cline{13-16}
\multicolumn{1}{|l|}{}                       & \multicolumn{1}{l|}{Clustering} & 0.86                 & 0.74                 & 0.65                 & 0.42                 &                      & 0.76                 & 0.69                 & 0.47                 & 0.42                 &                      & 0.83                 & 0.68                 & 0.51                 & 0.45                  \\ \cline{2-6} \cline{8-11} \cline{13-16}
\multicolumn{1}{|l|}{}                       & \multicolumn{1}{l|}{Embedder}   & \cellcolor{mygray}{0.33}        & \cellcolor{mygray}{0.36}        & \cellcolor{mygray}{0.47}        & \cellcolor{mygray}{0.55}        &                      & \cellcolor{mygray}{0.55}        & \cellcolor{mygray}{0.68}        & \cellcolor{mygray}{0.90}        & \cellcolor{mygray}{0.88}        &                      & \cellcolor{mygray}{0.85}        & 0.95                 & 0.96                 & 0.96                  \\
\multicolumn{1}{|l|}{}                       & \multicolumn{1}{l|}{Spring}     & 0.12                 & 0.14                 & 0.23                 & 0.30                 &                      & 0.24                 & 0.41                 & 0.61                 & 0.63                 &                      & 0.62                 & 0.72                 & 0.81                 & 0.82                  \\
\multicolumn{1}{|l|}{}                       & \multicolumn{1}{l|}{HyperMapCN} & 0.30 & 0.32 &0.40 & 0.44 &                      & 0.52 & 0.61 & 0.75 & 0.74 &                      & —\tnote{*} & —\tnote{*} & —\tnote{*} & —\tnote{*} \\
\multicolumn{1}{|l|}{}                       & \multicolumn{1}{l|}{Coalescent} & 0.32                 & 0.32                 & 0.35                 & 0.45                 &                      & 0.49                 & 0.40                 & 0.67                 & 0.76                 &                      & 0.66                 & \cellcolor{mygray}{0.96}        & \cellcolor{mygray}{0.97}        & \cellcolor{mygray}{0.96}         \\
\multicolumn{1}{|l|}{}                       & \multicolumn{1}{l|}{Poincare}   & 0.24                 & 0.22                 & 0.28                 & 0.37                 &                      & 0.29                 & 0.43                 & 0.53                 & 0.53                 &                      & 0.48                 & 0.49                 & 0.48                 & 0.45                  \\
\multicolumn{1}{|l|}{}                       & \multicolumn{1}{l|}{CHM}        & 0.38 & 0.38 & 0.39 & 0.42 &                      & 0.43 & 0.46 & 0.49 & 0.51 &                      & 0.51 & 0.55 & 0.72 & 0.71 \\
\multicolumn{1}{|l|}{}                       & \multicolumn{1}{l|}{HME}        & \textbf{0.40} & \textbf{0.40} & \textbf{0.49} & \textbf{0.57} &                      & \textbf{0.61} & 0.67 & 0.87 & \textbf{0.90} &                      & 0.71 & 0.86 & 0.87 & 0.87 \\ \hline
\end{tabular}
  \begin{tablenotes}
    \footnotesize
    \item[*] HyperMapCN did not finish running within 45 days for $C=-2$.
  \end{tablenotes}
\caption{HD-correlation on synthetic networks. We generate different synthetic networks by combing varying parameters of the hyperbolic random graph wherein T controls the clustering coefficient of networks and C tunes node's average degree. For each method, we calculate the Pearson correlation between the original hyperbolic random graph and the inferred embedding. It should be noted that among all the baseline methods, poincare disk model is considered as the latent hyperbolic geometry model while the Poincare Ball Model is employed in Poincare method.}
\label{HDcorrelation}
\end{threeparttable}
\end{table*}

\subsection{Evaluation on Real World Networks}

\begin{table*}[]\small
\begin{tabular}{|c|c|c|c|c|c|c|c|c|c|c|c|c|c|}
\hline
 & \multicolumn{3}{c|}{Routers} & \multicolumn{4}{c|}{C.Elegans} & \multicolumn{3}{c|}{Dros .Mel.} & \multicolumn{3}{c|}{SacchPomb} \\ \cline{2-14} 
\multirow{-2}{*}{Model} & IPv4 & IPv6 & Avg & 1 & 2 & 3 & Avg & 1 & 2 & Avg & 1 & 2 & Avg \\ \hline
\multicolumn{1}{|l|}{Common Neighbor} & \multicolumn{1}{l|}{0.892} & \multicolumn{1}{l|}{0.885} & \multicolumn{1}{l|}{0.889} & \multicolumn{1}{l|}{0.714} & \multicolumn{1}{l|}{0.741} & \multicolumn{1}{l|}{0.820} & \multicolumn{1}{l|}{0.758} & \multicolumn{1}{l|}{0.801} & \multicolumn{1}{l|}{0.783} & \multicolumn{1}{l|}{0.792} & \multicolumn{1}{l|}{0.660} & \multicolumn{1}{l|}{0.763} & \multicolumn{1}{l|}{0.712} \\ \hline
\multicolumn{1}{|l|}{Jaccard Coefficient} & \multicolumn{1}{l|}{0.876} & \multicolumn{1}{l|}{0.830} & \multicolumn{1}{l|}{0.853} & \multicolumn{1}{l|}{0.702} & \multicolumn{1}{l|}{0.722} & \multicolumn{1}{l|}{0.797} & \multicolumn{1}{l|}{0.740} & \multicolumn{1}{l|}{0.790} & \multicolumn{1}{l|}{0.776} & \multicolumn{1}{l|}{0.783} & \multicolumn{1}{l|}{0.660} & \multicolumn{1}{l|}{0.757} & \multicolumn{1}{l|}{0.708} \\ \hline
\multicolumn{1}{|l|}{Adamic/Adar} & \multicolumn{1}{l|}{0.895} & \multicolumn{1}{l|}{0.891} & \multicolumn{1}{l|}{0.893} & \multicolumn{1}{l|}{0.714} & \multicolumn{1}{l|}{0.744} & \multicolumn{1}{l|}{0.832} & \multicolumn{1}{l|}{0.763} & \multicolumn{1}{l|}{0.804} & \multicolumn{1}{l|}{0.785} & \multicolumn{1}{l|}{0.794} & \multicolumn{1}{l|}{0.661} & \multicolumn{1}{l|}{0.767} & \multicolumn{1}{l|}{0.714} \\ \hline
node2vec & 0.928 & 0.915 & 0.921 & 0.812 & 0.811 & 0.812 & 0.812 & 0.850 & 0.870 & 0.860 & 0.881 & 0.983 & \cellcolor[HTML]{C0C0C0}0.932 \\ \hline
deepwalk & 0.867 & 0.857 & 0.862 & 0.822 & 0.813 & 0.815 & \cellcolor[HTML]{C0C0C0}0.817 & 0.814 & 0.782 & 0.798 & 0.740 & 0.954 & 0.847 \\ \hline
LINE & 0.869 & 0.855 & 0.862 & 0.717 & 0.726 & 0.726 & 0.737 & 0.821 & 0.819 & 0.820 & 0.733 & 0.917 & 0.825 \\ \hline
HyperMap & 0.828 & 0.832 & 0.829 & 0.807 & 0.805 & 0.806 & 0.806 & 0.866 & 0.863 & 0.865 & 0.793 & 0.795 & 0.794 \\ \hline
CHM & 0.676 & 0.684 & 0.680 & 0.548 & 0.545 & 0.547 & 0.547 & 0.636 & 0.635 & 0.635 & 0.703 & 0.704 & 0.703 \\ \hline
SpringEmb & 0.548 & 0.549 & 0.548 & 0.606 & 0.608 & 0.606 & 0.607 & 0.726 & 0.726 & 0.726 & 0.685 & 0.684 & 0.684 \\ \hline
Poincare & 0.599 & 0.600 & 0.599 & 0.627 & 0.627 & 0.625 & 0.626 & 0.647 & 0.645 & 0.646 & 0.584 & 0.585 & 0.584 \\ \hline
Coalescent & 0.809 & 0.796 & 0.803 & 0.594 & 0.588 & 0.602 & 0.594 & 0.786 & 0.783 & 0.785 & 0.807 & 0.816 & 0.812 \\ \hline
HME & \multicolumn{3}{c|}{\cellcolor[HTML]{C0C0C0}0.931} & \multicolumn{4}{c|}{0.816} & \multicolumn{3}{c|}{\cellcolor[HTML]{C0C0C0}0.871} & \multicolumn{3}{c|}{0.925} \\ \hline
\end{tabular}
\caption{Link prediction result.}
\label{real-LP-AUC}
\end{table*}

\section{Related Work}
In this section, We review and differentiate the related work, including: network embedding, hyperbolic embedding and community detection.

\subsection{Network Embedding}
Network embedding aims to map each node in network to a continuous and distributed coordinate in a latent space. The embedding coordinates could reduce the sparsity and noise representation of conventional adjacency matrix. Various previous researches have shown its effectiveness and efficiency in several network tasks, including: node classification, link prediction and community detection\cite{cui2018survey}.

Inspired by word2vec\cite{mikolov2013efficient}, random walk based approaches are  proposed. Deepwalk\cite{perozzi2014deepwalk} uses node sequence generated by random walk as a equivalent of sentence. Node2vec\cite{grover2016node2vec} proposes more flexible explore strategies for random walkers. Consequently, more diverse neighbor sequence can be generated.

Matrix factorization is another popular and effective technique employed in network embedding. The representative work in this category is Grarep\cite{cao2015grarep}, which captures the different step relationship by performing matrix factorization on different transition matrices.

Another line of work computes the node embedding by preserving proximity between nodes up to different orders. LINE\cite{tang2015line} preserves the first-order and second-order proximity by approximating the empirical distribution of edge and neighbor node with sigmoid and softmax distributions. AROPE\cite{zhang2018arbitrary} proposes a arbitrary order proximity preserved embedding method based on SVD framework.

Besides the local structure information. Statistical property preserving network embedding has also attracted lots of attentions. \cite{wang2017community} preserves both the mesoscopic community structure of network, one of the most prominent property of network, and the microscopic structure features. \cite{feng2018representation} computes the node embedding while preserving the ubiquitous macroscopic scale-free property. \cite{lai2017prune} learns node embeddings from not only community-aware proximity but also global node ranking.

However, all above embedding methods do not consider the multiplexity of network.
Due to the complementary while redundant information of multiplex network, simple concatenation or summation of embedding is not feasible. Researchers propose severa methods for multiplex network. \cite{zhang2018scalable} learns each node a shared common embedding and a additional embedding for each type of relation. \cite{qu2017attention} employs the attention mechanism to infer the robust node representations across multiple view. In this work, we propose a multiplex network embedding method that maps each node into a latent hyperbolic space while preserving the multiplex community structure.

\subsection{Hyperbolic Embedding}
The seminal work that assumes hyperbolic geometry underlies the network is \cite{krioukov2010hyperbolic}. However, the plausibility that network nodes exist in hidden metric spaces is proved in \cite{serrano2008self}. With the underlying hyperbolic space assumption, scale free and strong clustering emerge naturally as the negative curvature geometry properties. Since then, the first and second order proximity based hyperbolic embedding approaches are proposed\cite{papadopoulos2015network}\cite{papadopoulos2015networkcommon}. The latent hyperbolic coordinates employed on AS (Autonomous System) Internet topology reveal interesting and promising result. The ASs locates at smaller angular distances tend to be geographically closer. \cite{blasius2018efficient} improves the embedding efficiency to quasilinear runtime. \cite{alanis2016efficient} employs the laplacian to infer the hyperbolic coordinates instead of the conventional proximity preservation. This series of work embed each node with the Poincar\'e disk model. \cite{nickel2017poincare} and \cite{balavzevic2019multi} propose to embed hierarchical data into Poincar\'e ball model instead, as it is well-suited for gradient-based optimization. \cite{wang2016hyperbolic} and \cite{faqeeh2018characterizing} find that the latent hyperbolic coordinates of nodes within one community tend to be closer. Then, this finding is employed to initialize the nodes embedding which could greatly improve the algorithm\cite{wang2016hyperbolic}. However, most of previous work focus on the single layer network embedding.

Although several researches reveal the hidden correlations between the hyperbolic coordinates of multiplex network\cite{kleineberg2016hidden}\cite{kleineberg2017geometric}. None of previous hyperbolic embedding work has employed these findings. We in this work extend the findings in \cite{wang2016hyperbolic} and \cite{faqeeh2018characterizing} to multiplex community.

\subsection{Community Detection}
Community is a group of vertices with prominent higher density within group connections than between group connections. Community detection is the one of the most concerned issues in network analysis. As for single layer networks, modularity maximization based work\cite{newman2006modularity} is one of the most representative methods. \cite{zhang2017modularity}\cite{wilson2017community} generalizes the multilayer modularity to adapt to networks with multiple layers. SBM (Stochastic Block Mode) based generative model is another important community detection method. \cite{rosvall2008maps} is a very different approach that detect community structure by encoding the flows of random walker. This work is generalized to multiplex scenario by relax the random walker jump freely between layers\cite{de2015identifying}. Recent work on community detection enhanced network embedding\cite{cavallari2017learning}\cite{wang2017community}\cite{lai2017prune}\cite{tu2018unified} are close to this work. However, none of them considers the multiple relations between nodes.

\section{Conclusion}

In this work, we have studied the problem of multiplex network embedding while preserving the multiplex community structure. Each node in each layer is mapped into a separate hyperbolic space. We jointly perform the multiplex community detection and hyperbolic network embedding by enforcing the random walker traverse in the latent hyperbolic spaces. We evaluate our method on several network tasks with various real world multiplex network datasets. The experiment results demonstrate that our method indeed improve both these two tasks simultaneously.

\bibliographystyle{ACM-Reference-Format}
\bibliography{ref}


\end{document}